\documentclass[12pt]{iopart}
\usepackage{iopams}
\usepackage{amsopn}
\usepackage{setstack}
\usepackage{bbm}
\usepackage[german,english]{babel}  
\usepackage[T1]{fontenc}
\usepackage{url} 
\usepackage{pst-uml}
\usepackage{dsfont}
\usepackage{amsthm}
\usepackage[latin9]{inputenc}
\newcommand{\ud}{\mathrm{d}}
\newcommand{\ui}{\mathrm{i}}
\newcommand{\ue}{\mathrm{e}}
\newcommand{\vl}{\boldsymbol{l}}
\newcommand{\gz}{{\mathbb Z}}
\newcommand{\rz}{{\mathbb R}}
\newcommand{\nz}{{\mathbb N}}
\newcommand{\kz}{{\mathbb C}}
\newcommand{\eins}{\mathds{1}}
\newcommand{\lf}{\mathfrak{l}}

\newcommand{\dk}{\frac{\ud\phantom{k}}{\ud k}}
\newcommand{\cL}{\mathcal{L}}
\newcommand{\cH}{\mathcal{H}}
\newcommand{\scp}{\left<\cdot,\cdot\right>_{\cH}}

\newcommand{\Th}{{\Theta}_{\Delta^{\frac{1}{2}}}}
\newcommand{\Zh}{{\zeta}_{\Delta^{\frac{1}{2}}}}
\newcommand{\tZh}{{\widetilde{\zeta}}_{\Delta^{\frac{1}{2}}}}

\newcommand{\T}{{\Theta}_{\Delta}}

\newcommand{\Zl}{{\zeta}_{\Delta}}

\newcommand{\Df}{\mathfrak{D}}

\newcommand{\dz}{d}

\newcommand{\be}{\begin{equation}}
\newcommand{\ee}{\end{equation}}
\newcommand{\eq}{\eqalign}
\newcommand{\Tho}{\Th^\mathrm{po}}

\newcommand{\lfp}{\lf_{p,n}}
\DeclareMathOperator{\kei}{kei}
\DeclareMathOperator{\im}{Im}
\DeclareMathOperator{\re}{Re}

\DeclareMathOperator{\Orr}{O}
\DeclareMathOperator{\orr}{o}
\DeclareMathOperator{\Lz}{L}
\DeclareMathOperator{\In}{I}
\DeclareMathOperator{\Ee}{E_1}
\DeclareMathOperator{\Hz}{H}
\DeclareMathOperator{\Dz}{D}
\DeclareMathOperator{\Pz}{P}

\DeclareMathOperator{\Pp}{PP}
\DeclareMathOperator{\Fp}{FP}
\DeclareMathOperator{\res}{Res}
\DeclareMathOperator{\Es}{E^{\ast}}

\DeclareMathOperator{\po}{\mathrm{po}}
\DeclareMathOperator{\lead}{lead}
\newtheorem{theorem}{Theorem}[section]

\newtheorem{prop}[theorem]{Proposition}
\newtheorem{cor}[theorem]{Corollary}

\newtheorem{rem}[theorem]{Remark}

\begin{document}
\title[An infinite quantum graph]{An exact trace formula and zeta functions for an infinite quantum graph with a non-standard Weyl asymptotics}
\author{Sebastian Egger né Endres}
\address{Institut f{\"u}r Theoretische Physik, Universit{\"a}t Ulm\newline Albert-Einstein-Allee 11, 89081 Ulm, Germany}
\ead{sebastian.endres@uni-ulm.de}
\author{Frank Steiner}
\address{Institut f{\"u}r Theoretische Physik, Universit{\"a}t Ulm\newline Albert-Einstein-Allee 11, 89081 Ulm, Germany \\ \vspace{2mm} and  \\ \vspace{2mm}
Centre de Recherche Astrophysique de Lyon, Universit\'{e} Lyon 1, CNRS UMR 5574, Observatoire de Lyon, 9 avenue Charles Andr\'{e}, 69230 Saint-Genis-Laval, France}
\ead{frank.steiner@uni-ulm.de}
\begin{abstract}
We study a quantum Hamiltonian that is given by the (negative) Laplacian and an infinite chain of $\delta$-like potentials with strength $\kappa>0$ on the half line $\rz_{\geq0}$ and which is equivalent to a one-parameter family of Laplacians on an infinite metric graph. This graph consists of an infinite chain of edges with the metric structure defined by assigning an interval $I_n=[0,l_n]$, $n\in\nz$, to each edge with length $l_n=\frac{\pi}{n}$. We show that the one-parameter family of quantum graphs possesses a purely discrete and strictly positive spectrum for each $\kappa>0$ and prove that the Dirichlet Laplacian is the limit of the one-parameter family in the strong resolvent sense. The spectrum of the resulting Dirichlet quantum graph is also purely discrete. The eigenvalues are given by $\lambda_n=n^2$, $n\in\nz$, with multiplicities $d(n)$, where $d(n)$ denotes the divisor function. We thus can relate the spectral problem of this infinite quantum graph to Dirichlet's famous divisor problem and infer the non-standard Weyl asymptotics $\mathcal{N}(\lambda)=\frac{\sqrt{\lambda}}{2}\ln\lambda +\Or\left(\sqrt{\lambda}\right)$ for the eigenvalue counting function. Based on an exact trace formula, the Vorono\"i summation formula, we derive explicit formulae for the trace of the wave group, the heat kernel, the resolvent and for various spectral zeta functions. These results enable us to establish a well-defined (renormalized) secular equation and a Selberg-like zeta function defined in terms of the classical periodic orbits of the graph, for which we derive an exact functional equation and prove that the analogue of the Riemann hypothesis is true.
\end{abstract}
\submitto{\JPA}
\maketitle
\section{Introduction}
\label{s1}
Quantum graphs are the subject of steadily increasing  interest since Roth \cite{Roth:1984} derived a trace formula for the heat kernel of the Laplacian with Kirchhoff boundary conditions on compact graphs. Kottos and Smilansky \cite{KottosSmilansky:1998} introduced compact quantum graphs as a model for quantum chaos and developed a more general trace formula for the Laplacian with vertex boundary conditions which depend at each vertex $v_l$ of the graph on a real parameter $\kappa_l$ (these boundary conditions are also called $\delta$-type boundary conditions \cite{Kuchment:2004,Albeverio:1988}). The choice $\kappa_l=0$ at each vertex corresponds to Kirchhoff boundary conditions (Kottos and Smilansky called them Neumann boundary conditions [see also \cite{Gnutzman:2006}]), and the choice $\kappa_l=\infty$ can be identified with Dirichlet boundary conditions. In the case where only two edges meet at a vertex $v_l$, these boundary conditions can be interpreted as being due to a ``$\delta$-like potential'' with strength $\kappa_l$ at the vertex $v_l$. We shall be faced with this situation in  section \ref{1} where we identify our Dirichlet quantum graph as such a limit of an infinite quantum graph with $\delta$-type boundary conditions at the vertices. Furthermore, in \cite{KottosSmilansky:1998} the spectral form factor and the level spacing distribution were investigated analytically and numerically and compared with the corresponding quantities of random matrix theory searching for ``characteristic signals'' of quantum chaos. Their analysis of the spectrum was continued in \cite{Berkolaiko:1999,Berkolaiko:2003,B:2003} by a sophisticated application and evaluation of the trace formula. Kostrykin and Schrader \cite{KostrykinSchrader:1999} presented a classification of all self adjoint characterizations of the Laplacian on compact graphs. The trace formula of Kottos and Smilansky was rigorously proved in \cite{Kurasov:2005} for $k$-independent boundary conditions and in \cite{BE:2008} for all self adjoint boundary conditions for the Laplacian on compact graphs by two different methods. For an excellent review on quantum graphs, in particular concerning the analysis of spectral properties like the level spacing distribution, $n$-point correlation functions or the form factor, see \cite{Gnutzman:2006}. 

However, much less is known for noncompact quantum graphs, in particular for quantum graphs with an infinite number of edges. One of the first researchers who studied Schrödinger operators on an infinite quantum graph (infinite number of edges) were Naimark and Solomyak \cite{Solomyak:2000} who proved a Weyl-type asymptotics for the Laplacian on a tree with a Dirichlet boundary condition at the root vertex and Kirchhoff boundary conditions at the remaining vertices. Solomyak \cite{Solomyak:2002} generalized this result to graphs with finite total length. Carlson \cite{Carslon:2000} was the first to prove that Laplacians with Kirchhoff boundary conditions at the vertices of graphs possessing finite total length (i.e. possessing the finite volume property) have compact resolvents and therefore possess purely discrete spectra.

In this paper we study a one-parameter family of infinite quantum graphs which do not satisfy the finite volume property.  We shall show that our quantum graph family possesses for each positive parameter $\kappa>0$ a purely discrete spectrum. Furthermore, we investigate the  Dirichlet quantum graph in detail which is the limit in the strong resolvent sense for letting the parameter $\kappa$ go to infinity. 

Kuchment \cite{Kuchment:2004} found a characterization of self adjoint realizations of the Laplacian by local boundary conditions in terms of the corresponding sesquilinear form of the Laplacian. However, he made the assumption that the lengths of the edges are uniformly bounded from below, which doesn't hold for our quantum graph. One of the most recent publications which deals with Schrödinger operators on infinite graphs is \cite{Solomyak:2009}. Therein, the number of negative eigenvalues of the Laplacian on a metric tree with a Dirichlet or Neumann boundary condition at the root and perturbed by a nonnegative potential was estimated.

Our infinite Dirichlet quantum graph system possesses as classical counterpart an integrable system, in contrast to the above mentioned systems. The quantal energy spectrum of our quantum graph is given by $\lambda_n=n^2$, $n\in\nz$, with multiplicity $d(n)$, where $d(n)$ denotes the divisor function. Thus, this quantum graph is closely related to the Dirichlet divisor problem \cite{Dirichlet:1849}, an old and intensively studied problem in number theory. Our special quantum graph system allows us to derive explicit expressions for many interesting quantities such as the trace of the wave group, the heat kernel, the resolvent and for various spectral zeta functions. 

It is remarkable that one can define for this infinite quantum graph a Selberg-like zeta function $Z(s)$ in terms of the length spectrum of the classical periodic orbits which possesses a functional equation and satisfies the analogue of the Riemann hypothesis. Furthermore, it turns out that this quantum system possesses a non-standard Weyl asymptotics for the spectral counting function $\mathcal{N}(\lambda)$ (see \eref{11c}) whose leading term agrees after an obvious rescaling of the edge lengths by a factor $\frac{1}{2\pi}$ with the leading term of the counting function for the nontrivial Riemann zeros. This is interesting for the search of the highly desired Hilbert-Polya operator, i.e. an operator which yields as eigenvalues or as wavenumbers the nontrivial Riemann zeros, and suggests that one should consider infinite quantum graphs for this purpose (see \cite{ES:2009} in this context for a discussion of finite quantum graphs equipped with the Berry-Keating operator and for references on the Hilbert-Polya operator).  

The Vorono\"i summation formula \eref{13}, well-known in number theory, plays for our physical system the role of a trace formula. The l.h.s. of the Vorono\"i summation formula \eref{13} corresponds to the quantum mechanical energy spectrum of the graph, while the r.h.s. can be split up in a ``Weyl term'' which corresponds to the three leading asymptotic terms of $N(x)$ in \eref{11}, and a ``periodic orbit term'' which invokes the geometric properties in terms of the classical periodic orbits of the graph.

Our Dirichlet quantum graph is the counterpart of the flat torus model $\left(-\Delta,\mathbb T^2\right)$ which consists of the Laplacian $-\Delta$ acting on a flat torus $\mathbb T^2:=\rz^2/\{L\gz\times L\gz\}$ characterized by a length scale $L>0$ equipped with periodic boundary conditions (see e.g. \cite{Steiner:2009}). For this quantum system the multiplicity of the $n^{th}$ eigenvalue is exactly given by the sum of squares function 
\be
\label{z109a}
r(n):= \#\left\{(m_1,m_2)\in\gz\times\gz,\quad n=m_1^2+m_2^2\right\},\quad n\in\nz.
\ee
Thus, the corresponding eigenvalue counting function is related to Gauss's circle problem, and it turns out that the formulae of Hardy \cite{Hardy:1915a,Hardy:1924,Hardy:1924a}, Landau \cite{Landau:1927a} and Vorono\"i \cite{Voronoi:1904a} play the role of the corresponding trace formula (see \cite{Steiner:2009}). However, its ``periodic orbit term'' involves a $J_0$ Bessel function in contrast to the $K_0$ and $Y_0$ Bessel functions in the trace formula \eref{13} for our Dirichlet quantum graph. 

Our paper is organized as follows. In section \ref{1} we define the quantum mechanical setting of the system and explain the connection to a one-parameter family of infinite quantum graphs, in particular we realize our Dirichlet quantum graph as a specific ``extremal'' point in this quantum graph family. Moreover, we prove that for each parameter $\kappa>0$ the spectrum of the quantum graph is purely discrete and present a necessary condition for the existence of an eigenfunction of the eigenvalue problem. In section \ref{z15} we introduce the Dirichlet quantum graph and show that the spectral multiplicities of its eigenvalues are given by the divisor function. 

The connection to the Dirichlet divisor problem and the Vorono\"i summation formula \eref{13} is made in section \ref{9}. We then proceed by examining the trace of the wave group $\Th(t)$ for this system, in particular we derive its small-$t$ asymptotics. In section \ref{30a} we calculate the generalized spectral zeta function $\Zh(s,k)$, similarly as in \cite{Steiner:1987}, which is the Mellin-Laplace transform of the trace of the wave group and investigate the poles of its meromorphic extension. Based on the analysis in section \ref{30a}, in particular due to the existence of a pole of $\Zh(s,k)$ at $s=1$, which corresponds to the trace of the resolvent, we define in section \ref{200} the ''regularized trace'' of the resolvent $\tau(k)$ and determine its asymptotics. We refine the analysis of section \ref{200} in section \ref{600}  by investigating the ``regularized spectral zeta function'' $\tZh(s,k)$. It turns out that the finite part of this function at $s=1$ exhibits a similar form as the finite part of the Hurwitz zeta function. This leads to the definition of the generalized gamma function $\widetilde{\Gamma}(k)$ whose poles are determined by the quantum mechanical wavenumbers of our graph. Furthermore, we derive a Stirling-like formula in section \ref{a1} for the generalized gamma function which is reminiscent of the Stirling formula for the ordinary gamma function. We discuss the trace of the ``generalized heat kernel'' ${\T}(t,k)$ and determine its asymptotics for $t\rightarrow0$ in section \ref{32aaaaa}. Additionally, we calculate the principal parts of the corresponding spectral zeta function ${\Zl}(s,k)$ which is now the Mellin-Laplace transform of the trace of the generalized heat kernel, and ascertain the domain of definition of the meromorphic extension of it. Section \ref{55} is devoted to the discussion of the problem of defining a secular equation for the wavenumbers in terms of the vertex scattering matrix according to the construction in \cite{KostrykinSchrader:2006b}. We show that this problem can be solved by truncating the quantum graph at an arbitrary vertex number, building the secular equation for the truncated graph according to \cite{KostrykinSchrader:2006b} and then performing a certain limit process for this secular equation by increasing the number at which the graph is truncated. This process leads to a Weierstraß product $D\left(k^2\right)$ in terms of the wavenumbers.

Furthermore, in section \ref{t1} we introduce the functional determinant  
\be
\label{e1}
\det\left(-\Delta+k^2\right):=\exp\left(\left.-\frac{\partial{\Zl}(s,k)}{\partial s}\right|_{s=0}\right)
\ee
of the unbounded operator $-\Delta+k^2$. This definition for a determinant of an operator was first introduced in \cite{Singer:1973} and applied to quantum field theory in \cite{Hawking:1977}. For a finite dimensional matrix this definition coincides with the standard determinant of the matrix. We show that the functional determinant $\det\left(-\Delta+k^2\right)$ agrees with the Weierstraß product $D\left(k^2\right)$. For a recent thorough discussion of zeta functions for finite quantum graphs with general self adjoint boundary conditions and their relation to the functional determinant see \cite{Harrison:2009} and \cite{Kirsten:2010}.
In section \ref{32aaa} we calculate the trace  $T(k)$ of the resolvent of $-\Delta$ and determine its asymptotics for $k\rightarrow0$ and $k\rightarrow\infty$ by an application of the trace formula. We proceed by constructing in section \ref{ar1} the ''periodic orbit zeta function'' $Z(s)$ in accordance with \cite{Steiner:1987} and \cite{Steiner:1991a} which resembles in many respects the famous Selberg zeta function for the hyperbolic Laplacian acting on compact Riemann surfaces. Furthermore, we define the function 
\be
\label{e3}
Z_\Delta(k):=\ue^{\ui\pi N^W(k)}Z(\ui k),
\ee
whose zeros are all symmetrically located on the ``critical line'' $\ui\rz$ and which takes real values on the critical line. In addition it is shown that the ``Weyl term'' of its zero counting function (which coincides with the ``Weyl term'' of the spectral counting function) is given by the phase of the first factor on the r.h.s. in \eref{e3}. These properties are shared by the well-known Hardy $Z$-function (see e.g. \cite[p.\,85]{Karatsuba:1992}),
%
%
%
which is related to the Riemann zeta function and is discussed in section \ref{ar1}. Additionally, we prove that the spectral counting function $N(x)$ of our Dirichlet quantum graph can be obtained by an adiabatic limit process in terms of the quantity $\frac{1}{\pi}\im\ln Z(s)$. Finally, we determine the asymptotics of $\ln\det\left(-\Delta+k^2\right)$ for $k\rightarrow 0$ and $k\rightarrow\infty$ in section \ref{500aaa} by applying results of section \ref{ar1}. 
\section{The quantum mechanical setting and a one-parameter family of Laplacians}
\label{1}
We first define a one-parameter family of quantum graphs and then identify our Dirichlet quantum graph, which is the subject of investigation in the following sections, as a certain limit of this family. 
We consider a single spinless particle of mass $m=\frac{1}{2}$ on the half line $\rz_{\geq0}$ subjected to a ``hard wall reflection'' at the  point $x_0=0$ and to a $\delta$-like potential with strength $\kappa\geq0$ at each point $x_n:=\pi\sum\limits_{m=1}^{n}\frac{1}{m}$, $n\in\nz$. Formally, the Schrödinger operator for this single particle system can be written as ($\hbar=2m=1$)
\be
\label{z1}
H=-\Delta+\kappa\sum\limits_{n=1}^{\infty}\delta\left(x-x_n\right),\quad\kappa\geq0,\quad x\geq0,
\ee
with $\Delta:=\frac{\ud^2\phantom{x}}{\ud x^2}$. In order to give a mathematically rigorous formulation of this system, we translate the hard wall reflection and the $\delta$-like potentials into local boundary conditions at the vertices of a quantum graph according to \cite{KottosSmilansky:1998,Kuchment:2004,Gnutzman:2006,Albeverio:1988}. The quantum graph $\mathfrak{G}$ consists of an infinite chain of consecutive edges $\{e_n\}_{n=1}^{\infty}$ with lengths $l_n=\frac{\pi}{n}$, $n\in\nz$, where the adjacent edges $e_{n}$ and $e_{n+1}$ are linked by a vertex $v_n$. Therefore, the total length $\cL:=\sum\limits_{n=1}^{\infty}\frac{\pi}{n}$ (also called the volume) of the quantum graph is infinite and $\mathfrak{G}$ is a noncompact quantum graph. The quantum mechanical Hilbert space $\left(\cH,\scp\right)$ for this system consists of the vector space
\begin{equation}
\label{2}
\cH:=\left\{\psi\in\bigoplus\limits_{n=1}^{\infty}\Lz^2\left[0,\frac{\pi}{n}\right];\quad \|\psi\|_{\cH}:=\left<\psi,\psi\right>_{\cH}^{\frac{1}{2}}<\infty\right\}
\end{equation}
equipped with the scalar product
\begin{equation}
\label{3}
\eqalign{
\left<\psi,\varphi\right>_{\cH} & :=\sum\limits_{n=1}^{\infty}\left<\psi_n,\varphi_n\right>_{\Lz^2\left[0,\frac{\pi}{n}\right]},\\ 
 & \psi=\bigoplus\limits_{n=1}^{\infty}\psi_n, \quad \varphi=\bigoplus\limits_{n=1}^{\infty}\varphi_n,\quad \psi_n,\varphi_n\in\Lz^2\left[0,\frac{\pi}{n}\right].
}
\end{equation}
In order to find the corresponding self adjoint Laplace operator $(\Delta,\Dz_\kappa(\Delta))$, in particular the domain of definition $\Dz_\kappa(\Delta)$ for the Hamilton operator \eref{z1}, we use the approach by sesquilinear forms of \cite{Kuchment:2004} for infinite quantum graphs. So far this approach is the only rigorous one which guarantees that the resulting operator is self adjoint. However, our edges $e_n$, $n\in\nz$, are not uniformly bounded from below by a strictly positive number which was assumed in \cite{Kuchment:2004,Albeverio:1988}. Therefore, we have in the following to generalize the setting by an additional requirement on the domain of the sesquilinear form.  

The boundary condition on the function $\psi\in\cH$ at the vertex $v_0$ corresponding to the hard wall reflection at the origin is given by a Dirichlet boundary condition \cite{KottosSmilansky:1998,Gnutzman:2006} $\psi_1(0)=0$. Due to the Hamiltonian \eref{z1}, the boundary condition at the $n^{th}$ vertex $v_n$  corresponds to (using the definition $\psi_{v_n}:=\psi_{n+1}(0)$)
\be
\label{z2}
\psi_{n}\left(\frac{\pi}{n}\right)=\psi_{n+1}(0)=\psi_{v_n}, \quad \psi'_{n+1}(0)-\psi'_{n}\left(\frac{\pi}{n}\right)=\kappa\psi_{v_n},\quad n\in\nz.
\ee
%
%
%
%
%
%
%
%
%
%
%
%
%
We have to incorporate these boundary conditions in the approach of \cite{Kuchment:2004} and then have to show that these boundary conditions together with some additional requirements define a self adjoint Laplacian on $\mathfrak{G}$. Due to the failure of a strictly positive lower bound for the edges we adapt the ``graph Sobolev function space'' in \cite{Kuchment:2004} to our infinite quantum graph $\mathfrak{G}$ by defining the modified Sobolev space
\be
\label{z5}
\fl
\eq{
&\Hz^1(\mathfrak{G}):=\\
&\left\{\psi\in\mathcal{H}; \ \psi_n\in \Hz^1\left(0,\frac{\pi}{n}\right), n\in\nz, \ \left\|\psi'\right\|_{\cH}<\infty, \ \sum\limits_{n=1}^{\infty}\left[\left|\psi_{n}(0)\right|^2+\left|\psi_{n}\left(\frac{\pi}{n}\right)\right|^2\right]<\infty\right\}.
}
\ee
In (\ref{z5}) we denote by $\psi_n$ the corresponding vector in the orthogonal decomposition of $\psi$ as in (\ref{3}), and $\Hz^1\left(0,\frac{\pi}{n}\right)$ is the Sobolev space of all functions in $\Lz^2\left[0,\frac{\pi}{n}\right]$ for which the first weak derivative is also an element of $\Lz^2\left[0,\frac{\pi}{n}\right]$ (see e.g.\ \cite{BE:2008,ES:2009}). Obviously, it holds $\Hz^1(\mathfrak{G})\subset C(\mathfrak{G})$, where  $C(\mathfrak{G})$ is the set of all continuous functions on $\mathfrak{G}$. We now define the sesquilinear form $h_\kappa$  conforming to our initial Hamiltonian \eref{z1} in accordance with \cite{Kuchment:2004,Gnutzman:2006} as
%
%
%
%
%
%
%
%
%
%
%
%
%
%
\be
\label{z6}
h_\kappa\left[\psi,\varphi\right]:=\left<\psi',\varphi'\right>_{\cH}+\kappa\sum\limits_{n=1}^{\infty}\overline{\psi_{v_n}}\varphi_{v_n},\quad \psi,\varphi\in D(h), \quad \kappa\geq0, 
\ee
where the domain of definition $D(h)$ consists of all functions $\psi\in \Hz^1(\mathfrak{G})$ fulfilling $\psi_1(0)=0,$ $\psi_{n}\left(\frac{\pi}{n}\right)=\psi_{n+1}(0)=\psi_{v_n}$ for all $n\in\nz$. 
%
%
The corresponding quadratic form $h_\kappa[\psi]:=h_\kappa[\psi,\psi]$ is positive and one can show applying the definition \eref{z5} of $\Hz^1(\mathfrak{G})$  that $h_\kappa$ is closed. Therefore, following the proofs in \cite{Kuchment:2004} there corresponds to each  $h_\kappa$, $\kappa\geq0$, a unique self adjoint Laplacian $(\Delta,\Dz_\kappa(\Delta))$ on $\mathfrak{G}$
\begin{equation}
\label{4}
\Delta\psi:=\bigoplus\limits_{n=1}^{\infty} \Delta_n\psi_n, \quad \psi\in \Dz_\kappa(\Delta),\quad  
\end{equation}
with domain of definition 
\be
\label{z7}
\hspace{-2cm}\eq{
\Dz_\kappa(\Delta)&:=\left\{\psi\in \Hz^1(\mathfrak{G}); \ \ \psi_n\in\Hz^2\left(0,\frac{\pi}{n}\right), \ \ \left\|\Delta\psi\right\|_{\mathcal{H}}<\infty, \ \ \psi_1(0)=0,\right. \\
& \hspace{1.0cm}\left.\psi_{n}\left(\frac{\pi}{n}\right)=\psi_{n+1}(0)=\psi_{v_n}, \ \  \psi'_{n+1}(0)-\psi'_{n}\left(\frac{\pi}{n}\right)=\kappa\psi_{v_n},\ \ n\in\nz \right\}.
}
\ee
In particular the required boundary conditions \eref{z2} are incorporated. In such a way one obtains a one-parameter family of self adjoint Laplacians $(\Delta,\Dz_\kappa(\Delta))$ acting on $\mathfrak{G}$. The case $\kappa=0$ corresponds to a Dirichlet boundary condition on the vertex $v_0$ and Kirchhoff boundary conditions (called Neumann boundary conditions in \cite{KottosSmilansky:1998}) on the remaining vertices. Hence, the Laplacian $(\Delta,\Dz_0(\Delta))$ is equivalent to the Laplacian $(\Delta,D_{\rz_{\geq0}})$ acting in $\Lz^2\left[0,\infty\right)$, where the domain of definition is given by $D_{\rz_{\geq0}}=\left\{\psi\in\Hz^2(0,\infty); \quad \psi(0)=0\right\}$. Therefore, $(\Delta,\Dz_0(\Delta))$ possesses a purely continuous spectrum $\sigma{(\Delta,\Dz_0(\Delta))}=[0,\infty)$. 

Since the quadratic form $h_\kappa$, $\kappa\geq0$, is positive, we infer that $(\Delta,\Dz_\kappa(\Delta))$ possesses no negative eigenvalues. Furthermore, since the general solution of the eigenvalue problem 
\be
\label{z9a}
-\Delta\varphi=\lambda\varphi, \quad \varphi\in\Dz_\kappa(\Delta), \quad \kappa>0,
\ee
is on the $n^{th}$ edge for $\sqrt{\lambda}=k^2=0$ of the form $\varphi_n(x)=a_nx+b_n$, we infer from the boundary conditions \eref{z7} by a straightforward calculation that zero is not an eigenvalue of $(\Delta,\Dz_\kappa(\Delta))$ for every $\kappa\geq0$. Moreover, since the general solution of the eigenvalue problem \eref{z9a} is for $\sqrt{\lambda}=k>0$ on the $n^{th}$ edge of the form $\varphi_n(x)=a_n\sin(kx+b_n)$, we infer again from the boundary conditions \eref{z7} that the solution of the eigenvalue problem \eref{z9a} must have the property that every component $\varphi_n$ in the decomposition \eref{3} is for every $k\in\rz$ unequal to the zero function.

Since the quantum graph $\mathfrak{G}$ consists of infinitely many edges, we cannot directly apply the results of \cite{KostrykinSchrader:1999,Kuchment:2004,Gnutzman:2006,KostrykinSchrader:2006,KostrykinSchrader:2006b,BE:2008,ES:2009} for a spectral analysis, in particular for an eigenvalue investigation . However, by a rearrangement of the rows and columns of the matrices therein some results still hold, but unfortunately not all. We shall discuss this in more detail in section \ref{55} for the Dirichlet quantum graph which corresponds to $\kappa=\infty$.

Now, we shall prove that $(\Delta,\Dz_\kappa(\Delta))$ possesses a purely discrete spectrum for every $\kappa>0$. We use a criterion of \cite[p.\,28,106]{Arendt:2006} for $(\Delta,\Dz_\kappa(\Delta)$ to possess a compact resolvent and therefore a purely discrete spectrum  \cite[p.\,6]{Arendt:2006}. In order to use this criterion, we need the following theorem.
\begin{theorem}
\label{z17}
The embedding
\be
\label{z10a}
(D(h),\left\|\cdot\right\|_{D(h)})\hookrightarrow\left(\cH,\scp\right)
\ee
with
\be
\label{z16}
\left\|\psi\right\|_{D(h)}:=\left[\left\|\psi\right\|^2_{\cH}+h_{\kappa}(\psi)\right]^{\frac{1}{2}},\quad \psi\in D(h),
\ee
is compact.
\end{theorem}
\begin{proof}
Let a sequence $(\psi^m)_{m\in\nz}$ with $\left\|\psi^m\right\|_{D(h)}\leq M$, $\psi^m\in D(h)$ for all $m\in\nz$, $M>0$ be given. We have to prove that there exists a subsequence $\left(\psi^{m_j}\right)_{j\in\nz}$ of $(\psi^m)_{m\in\nz}$ and a $\psi\in\cH$ with
\be
\label{z18}
\psi^{m_j}\rightarrow\psi,\quad j\rightarrow\infty\quad  \mbox{in} \quad \cH.
\ee
It follows immediately by the embedding theorem of Morrey \cite[p.\,118]{Dobrowolski:2006} that there exists for each $n\in\nz$ a $\psi_n\in C^{0,\alpha}\left[0,\frac{\pi}{n}\right]\subset\Lz\left[0,\frac{\pi}{n}\right]$, $0<\alpha<\frac{1}{2}$, ($C^{0,\alpha}\left[0,\frac{\pi}{n}\right]$ is the set of Hölder continuous functions on $\left[0,\frac{\pi}{n}\right]$ with exponent $\alpha$) such that for each $\left(\psi^m_n\right)_{m\in\nz}$ (the component of $\left(\psi^m\right)_{m\in\nz}$ on the $n^{th}$ edge) there exists a subsequence $\left(\psi^{m_l}_n\right)_{l\in\nz}$  with 
\be
\label{z17a}
\hspace{-2cm}\psi^{m_l}_n\rightarrow\psi_n, \quad l\rightarrow\infty \quad \mbox{in} \quad C^{0,\alpha}\left[0,\frac{\pi}{n}\right] \ \ \mbox{and in} \ \ \Lz^2\left[0,\frac{\pi}{n}\right], \ \ n\in\nz, \ \ 0<\alpha<\frac{1}{2}.
\ee
Thus, by Cantor's diagonal argument we can find a subsequence $\left(\psi^{m_j}\right)_{j\in\nz}$ of $(\psi^m)_{m\in\nz}$ for which every component $\psi^{m_j}_n$ fulfill \eref{z17a} . We define $\psi:=\oplus_{n=1}^{\infty}\psi_n$ with $\psi_n$ as in \eref{z17a}. We shall prove that it holds $\psi\in\cH$ and furthermore that $\left(\psi^{m_j}\right)_{j\in\nz}$ and $\psi$ fulfill \eref{z18}. 

Due to $\eref{z17a}$ the boundary values $\psi_n^{m_j}(0)$ and $\psi_n^{m_j}\left(\frac{\pi}{n}\right)$ converge to $\psi_n(0)$ and $\psi_n\left(\frac{\pi}{n}\right)$ for $j\rightarrow\infty$. Due to the continuity condition of $\psi^{m_j}$ for every $j\in\nz$ at the vertices $v_l$, $l\in\nz$, we deduce that $\psi$ is also continuous at these vertices. A similar argument yields $\psi_1(0)=0$. Thus, we define as above $\psi_{v_n}:=\psi_{n+1}(0)$, $n\in\nz$.

We shall first prove $\sum\limits_{n=1}^{\infty}\left|\psi_{v_n}\right|^2<\infty$.  For $0<l<p$ it follows by the assumption on $\psi^{m_j}$ and the triangle inequality that it holds for every $\epsilon>0$ ($j$ large enough, $\kappa>0$)
\be
\label{z19}
\left[\sum\limits^{p}_{n=l}\left|\psi_{v_n}\right|^2\right]^{\frac{1}{2}} \leq\left[\sum\limits^{p}_{n=l}\left|\psi_{v_n}-\psi^{m_j}_{v_n}\right|^2\right]^{\frac{1}{2}}+\left[\sum\limits^{p}_{n=l}\left|\psi^{m_j}_{v_n}\right|^2\right]^{\frac{1}{2}}\leq\epsilon+\frac{M}{\kappa}.
\ee
This proves $\sum\limits_{n=1}^{\infty}\left|\psi_{v_n}\right|^2<\infty$. 

From the proof of Morrey's inequality in Satz 6.26 of \cite[p.\,119]{Dobrowolski:2006} and from the prerequisite $\left\|\psi^{m_j}\right\|_{D(h)}\leq M$ for all $j\in\nz$, it follows that one can uniformly estimate in $n\in\nz$ (the constant $c>0$ doesn't grow with $n$)
\be
\label{z20}
\hspace{-2cm}\left|\psi^{m_j}_n(x)-\psi^{m_j}_n(y)\right|\leq c\left|x-y\right|^{\frac{1}{2}}\left\|
{\psi^{m_j}_n}'\right\|_{\Lz\left[0,\frac{\pi}{n}\right]}\leq cM \left|x-y\right|^{\frac{1}{2}},\quad x,y\in\left[0,\frac{\pi}{n}\right].
\ee  
Since $\psi^{m_j}_n(x)\rightarrow\psi_n(x)$ for $j\rightarrow\infty$ in $\kz$ (pointwise) for every $n\in\nz$ and fixed $x\in\left[0,\frac{\pi}{n}\right]$, we infer that the estimate \eref{z20} between the first and last term also holds for $\psi_n$, $n\in\nz$, (the middle term in \eref{z20} doesn't generally exist in this case).

We calculate using \eref{z20}, $1<l<p$, 
\be
\label{z21}
\hspace{-2cm}\eq{
\left\|\psi\right\|^2_{l,p}:=\sum\limits_{n=l}^{p}\int\limits_{0}^{\frac{\pi}{n}}\left|\psi_n(x)\right|^2\ud x&\leq\sum\limits_{n=l}^{p}\int\limits_{0}^{\frac{\pi}{n}}\left|\left|\psi_{v_n}\right|+cMx^{\frac{1}{2}}\right|^2\ud x\\
 &= \sum\limits_{n=l}^{p}\left[\frac{\pi}{n}\left|\psi_{v_n}\right|^2+\frac{4}{3}cM\left(\frac{\pi}{n}\right)^{\frac{3}{2}}\left|\psi_{v_n}\right|+\frac{1}{2}c^2M^2\left(\frac{\pi}{n}\right)^2\right]\\
 &=: C_{l,p}, 
}
\ee
where 
\be
\label{z22}
\lim_{p\rightarrow\infty}C_{l,p}=:C_{l} \ \ \mbox{exists} \ \ \mbox{and} \ \ C_{l}\rightarrow0, \quad l\rightarrow\infty.
\ee
This proves $\psi\in\cH$. 

By an analogous calculation for $\psi^{m_j}$ as in \eref{z1} and since $\left\|\psi^{m_j}\right\|_{D(h)}\leq M$ for all $j\in\nz$, we can deduce that it also holds (the constant $c$ in \eref{z20} can be chosen independently of $j\in\nz$ \cite[p.\,119]{Dobrowolski:2006})
\be
\label{z23}
\left\|\psi^{m_j}\right\|^2_{l,p}\leq\widetilde{C}_{l,p}, \quad \mbox{for all} \quad j\in\nz,\quad l<p,
\ee
where $\widetilde{C}_{l,p}$ fulfills also \eref{z22}. We then derive
\be
\label{z24}
\left\|\psi^{m_j}-\psi\right\|^2_{\cH}=\sum\limits_{n=1}^l\left\|\psi^{m_j}_n-\psi_n\right\|^2_{\Lz\left[0,\frac{\pi}{n}\right]}+C_{l}+\widetilde{C}_{l},\quad l>1.
\ee
Due to \eref{z17a} and \eref{z22} we infer $\psi^{m_j}\rightarrow\psi$ for $j\rightarrow\infty$ in $\cH$. This proves the theorem \ref{z17}.
\end{proof} 
%
From theorem \ref{z17} we finally obtain with \cite[p.\,6,28,106]{Arendt:2006} 
\begin{theorem}
\label{z25}
The infinite quantum graph $\mathfrak{G}$ defined by the Hamiltonian $(\Delta,\Dz_\kappa(\Delta))$, where $\Delta$ denotes the Laplacian $\eref{4}$ with domain of definition $D_{\kappa}(\Delta)$ defined in $\eref{z7}$, possesses for each $\kappa>0$ a compact resolvent and therefore a purely discrete, strictly positive spectrum.
\end{theorem}
Finally, we want to present a necessary quantization condition for the existence of an eigenfunction $\varphi_k =\oplus_{n=1}^{\infty}\varphi_{k,n}$ with $k\notin\left\{mn;\ mn\in\nz_0\right\}$ (see \cite{Kuchment:2004,Gnutzman:2006}) and
\be
\label{z10}
\hspace{-2.5cm}
\varphi_{k,n}(x) :=\frac{1}{\sin\left(k\frac{\pi}{n}\right)}\left(a_{k,n}\sin(kx)+b_{k,n}\sin\left(k\left(\frac{\pi}{n}-x\right)\right)\right),\quad x\in\left[0,\frac{\pi}{n}\right],\quad n\in\nz,  
\ee
of the eigenvalue problem \eref{z9a} with $\lambda=k^2$. (For each $k\in\left\{mn;\ mn\in\nz\right\}$ one only has to modify the following expressions on a finite set of edges, but this can be treated in a similar way). Requiring the boundary conditions \eref{z2} at all vertices of the graph $\mathfrak{G}$, we obtain the relations
\be
\label{z11}
\eq{
\phantom{aaaaaa} b_{k,1} &=0\\
\left(
\begin{array}{c}
a_{k,n+1} \\
b_{k,n+1}
\end{array}
\right)&=\left(
\begin{array}{cc}
A_n(k) & B_n(k)\\
1 & 0 
\end{array}
\right)
\left(
\begin{array}{c}
a_{k,n} \\
b_{k,n}
\end{array}
\right),
}
\ee
with $n\in\nz$, $k\notin\left\{mn;\ mn\in\nz_0\right\}$,
\be
\label{z12}
\fl A_n(k) :=\left[\frac{\kappa}{k}+\cot\left(k\frac{\pi}{n}\right)+\cot\left(k\frac{\pi}{n+1}\right)\right]\sin\left(k\frac{\pi}{n+1}\right), \quad B_n(k) :=-\frac{\sin\left(k\frac{\pi}{n+1}\right)}{\sin\left(k\frac{\pi}{n}\right)}.
\ee
Using $b_{k,n}=a_{k,n-1}$ the solutions of relation \eref{z11} are equivalent with the solutions of the three-term recursion relation
\be
\label{z13}
a_{k,n+1}=A_n(k)a_{k,n}+B_n(k)a_{k,n-1}, \quad n\geq1,\quad k\notin\left\{mn;\ mn\in\nz\right\},
\ee
with the starting values \eref{z7}
\be
\label{z14}
b_{k,1}:=a_{k,0}=0.
\ee
We conclude that every solution of \eref{z9a} for $\sqrt{\lambda}=k\notin\left\{mn;\ mn\in\nz\right\}$ is completely determined by the two conditions \eref{z13} and \eref{z14}.

It is worthwhile to note that the quadratic form $h_\kappa$ converges for $\kappa\rightarrow\infty$ ``from below'' \cite[p.\,461]{Kato:1980} to a quadratic form whose corresponding sequilinear form is given by
\be
\label{z8}
\hspace{-2cm}h[\psi,\varphi]=\left<\psi',\varphi'\right>_{\cH},\quad \psi,\varphi\in\Hz^1(\mathfrak{G}), \quad \psi_{n}(0)=\psi_n\left(\frac{\pi}{n}\right)=0,\quad n\in\nz.
\ee
Hence, we obtain for the related operators \cite[p.\,455]{Kato:1980} 
\be
\label{z9}
(\Delta,\Dz_\kappa(\Delta))\rightarrow \left(\Delta,\Dz(\Delta)\right),\quad\kappa\rightarrow\infty\quad \mbox{in the strong resolvent sense},
\ee
where on the r.h.s. the Dirichlet Laplacian is given with domain of definition 
\begin{equation}
\label{5}
\hspace{-2.3cm}\Dz(\Delta):=\left\{\psi\in\cH;\ \psi_{n}\in\Hz^2\left(0,\frac{\pi}{n}\right), \  \left\|\Delta \psi\right\|<\infty,
\ \psi_n(0)=\psi_n\left(\frac{\pi}{n}\right)=0, \ n\in\nz\right\}.
\end{equation}
%
That one can omit in \eref{5} the requirement $\left\|\psi'\right\|_{\cH}<\infty$ follows from a suitable Poincaré inequality. 

It would be very interesting to investigate the asymptotic behaviour of the spectrum $\sigma\left((\Delta,\Dz_\kappa(\Delta))\right)$ for $\kappa>0$ in the limit $\lambda\rightarrow\infty$, but we shall discuss this in a future publication.
%
\section{The Dirichlet quantum graph}
\label{z15}
We observe that for the Dirichlet quantum graph $(\Delta,\Dz(\Delta))$ due to the Dirichlet boundary conditions on all interval ends there is no ``interaction'' between the edges. Thus, the limit process $\kappa\rightarrow\infty$ for the one-parameter family $(\Delta,\Dz_\kappa(\Delta))$ ``isolates'' the edges of the graph $\mathfrak{G}$. The S-matrix $S$ in the sense of \cite{KottosSmilansky:1998,KostrykinSchrader:1999} of $\left(-\Delta,\Dz(\Delta)\right)$ corresponds to the identity operator $-\eins$ on $l^2$ being the set of all square summable sequences of real numbers. For compact quantum graphs with Dirichlet boundary conditions the assigned quantum graph as discussed in \cite{KostrykinSchrader:2006} (the unique quantum graph with maximal vertex numbers for which the S-matrix is local) is simply the set of edges with no common vertex for two different interval ends. This corresponds to a ``hard wall reflection'' of the particle at the interval ends (see \cite{ES:2009}).  The ``hard wall'' interpretation of the Dirichlet boundary conditions still holds for the infinite Dirichlet quantum graph. 

For this simple system the stationary Schrödinger equation
\begin{equation}
\label{6}
-\Delta\psi=\lambda\psi, \quad \psi\in \Dz(\Delta),
\end{equation}
is explicitly solvable. The normalized eigenvectors $\psi_{n,m}$ are expressed by the orthogonal decomposition as in (\ref{3}) by ($x\in\left[0,\frac{\pi}{l}\right]$)
\begin{equation}
\label{7}
\hspace{-2cm}\Pz_l\left({\psi_{n,m}}\right)(x) :=\delta_{ln}\sqrt{\frac{2n}{\pi}}\sin \left(k_{n,m}x\right)\quad \mbox{with}\quad k_{n,m}:=nm,\quad l,n,m\in\nz,
\end{equation}
where we have introduced the projector $\Pz_l$ from $\cH$ onto $\Lz^2\left[0,\frac{\pi}{l}\right]$ (see (\ref{3})) and the wavenumbers $k_{n,m}:=+\sqrt{\lambda_{n,m}}$. Since the set $\left\{\psi_{n,m};\quad m=1,\ldots,\infty\right\}$ ($n$ fixed) forms an orthonormal basis of $\Lz^2\left[0,\frac{\pi}{n}\right]$, we immediately infer that the eigenvectors $\psi_{n,m}$ form a complete orthonormal basis of $\cH$. Therefore, the spectrum $\sigma\left(\left(-\Delta,\Dz(\Delta)\right)\right)$ is purely discrete (see e.g.\ \cite{Reed:1980}) and the set of the corresponding wavenumbers is identical with the set of the natural numbers $\nz$. The multiplicity of the eigenvalue $\lambda=k^2$ respectively of the corresponding wavenumber $k$ is given by the divisor function $\dz(k)$ which counts the divisors of $k$, unity and $k$ itself included, i.e.\ it is defined as 
\begin{equation}
\label{8}
\dz(k):=\#\left\{(n,m)\in\nz\times\nz;\quad nm=k\right\}, \quad k\in\nz.
\end{equation}
Note that the divisor function $d(k)$, $d(1)=1$, $d(2)=2$, $d(3)=2$, $d(4)=3$, $d(5)=2$, $d(6)=4$, $\ldots$, with $d(p)=2$ for $p$ prime, is a very irregular function with asymptotic behaviour 
\begin{equation}
\label{9a}
d(k)=\Or\left(k^{\epsilon}\right),\quad k\rightarrow\infty \quad \mbox{for any} \quad \epsilon>0.
\end{equation}

\section{Spectral asymptotics and trace formula}
\label{9}
From section \ref{z15} and in particular (\ref{8}) we conclude that the number of wavenumbers of (\ref{6}) less than $x$, the spectral counting function $N(x)$, is closely related to the Dirichlet divisor problem. Explicitly, one obtains 
\begin{equation}
\label{10}
\hspace{-1.5cm}N(x):=\#\left\{(n,m)\in\nz\times\nz;\quad \sqrt{\lambda_{n,m}}=nm\leq x\right\}=\sum\limits_{n\leq x}\dz(n),\quad x\in\rz_{>0},
\end{equation}
where the sum on the r.h.s in \eref{10} is set to be zero for $x<1$. The famous divisor problem of Dirichlet \cite{Dirichlet:1849} is that of determining the asymptotic behaviour of $N(x)$ as $x\rightarrow\infty$. Reducing the problem to a lattice point problem, i.e.\ counting the positive integer lattice points under the hyperbola $n\leq\frac{x}{m}$, Dirichlet proved \cite{Dirichlet:1849}
\begin{equation}
\label{11}
\hspace{-1.5cm}N(x)=x\ln x +\left(2\gamma-1\right)x+\frac{1}{4}+N^{\mathrm{Osc.}}(x), \quad N^{\mathrm{Osc.}}(x)=\Or\left(x^\frac{1}{2}\right),\quad x\rightarrow\infty,
\end{equation}
wherein $\gamma$ denotes the Euler constant and $N^{\mathrm{Osc.}}(x)$ is, in physics notation, the ``oscillatory term'' to $N(x)$. Note that $N^{\mathrm{Osc.}}(x)$ is identical to $\Delta(x)$ in \cite{Wilton:1931a} and possesses jump discontinuities with jumps of magnitude $d(n)$ at $n\in\nz$. From (\ref{11}) it follows that the mean value of the multiplicities of the corresponding wavenumbers $k_{n,m}$ respectively the eigenvalues $\lambda_{n,m}$ of our infinite quantum graph is given by 
\begin{equation}
\label{11a}
\overline{d}(x):=\frac{1}{x}\sum\limits_{n^\leq x}d(n)=\ln x+(2\gamma-1)+\Or\left(\frac{1}{\sqrt{x}}\right),\quad x\rightarrow\infty.
\end{equation}
Furthermore, we deduce that the counting function of the eigenvalues $\lambda_{n,m}=k^2_{n,m}=n^2m^2$, 
\begin{equation}
\label{11b}
\mathcal{N}(\lambda):=\#\left\{(n,m)\in\nz\times\nz;\quad\lambda_{n,m}\leq\lambda\right\},
\end{equation}
is given by $\mathcal{N}(\lambda)=N\left(\sqrt{\lambda}\right)$, i.e.\ we derive from Dirichlet's result (\ref{11}) the modified Weyl's law
\begin{equation}
\label{11c}
\mathcal{N}(\lambda)=\frac{\sqrt{\lambda}}{2}\ln \lambda +(2\gamma-1)\sqrt{\lambda}+\Or\left(\lambda^\frac{1}{4}\right),\quad \lambda\rightarrow \infty.
\end{equation}
This should be compared with the counting function for a compact quantum graph with total length $\mathfrak{L}<\infty$, for which one has the standard Weyl asymptotics $\mathcal{N}(\lambda)=\frac{\mathfrak{L}}{\pi}\sqrt{\lambda}+\Or(1)$ \cite{BE:2008,ES:2009}.

Dirichlet's estimate \eref{11} on $N^{\mathrm{Osc.}}(x)$ was later improved e.g. by \cite{Voronoi:1903} ($N^{\mathrm{Osc.}}(x)=\Or\bigl(x^\frac{1}{3}\ln x\bigr)$), \cite{Vandercorput:1922} ($N^{\mathrm{Osc.}}(x)=\Or\left(x^{\Theta}\right)$ with $\Theta<\frac{33}{100}$) and \cite{Huxley:2003} ($N^{\mathrm{Osc.}}(x)=\Or\left(x^{\frac{131}{146}}\right)$). In \cite{Hardy:1917} it was proved that $N^{\mathrm{Osc.}}(x)=\Or\left(x^{\Theta}\right)$ with $\Theta>\frac{1}{4}$, and it is not unlikely (Hardy's conjecture) that
\be
\label{11ca}
N^{\mathrm{Osc.}}(x)=\Or\left(x^{\frac{1}{4}+\epsilon}\right),\quad x\rightarrow\infty,
\ee
for all positive $\epsilon$, but the exact order of $N^{\mathrm{Osc.}}(x)$ is still unknown. However, there exists an explicit expression for $N^{\mathrm{Osc.}}(x)$ (see e.g. \cite{Wilton:1931a})
\be
\label{11d}
\hspace{-2cm}N^{\mathrm{Osc.}}(x)=\frac{d(x)}{2}-\sqrt{x}\sum\limits_{n=1}^{\infty}\frac{d(n)}{\sqrt{n}}\left[\frac{2}{\pi}K_1\left(4\pi\sqrt{nx}\right)+Y_1\left(4\pi\sqrt{nx}\right)\right],\quad x\in\rz_{>0},
\ee
where we have defined
\be
\label{q11d}
d(x):=\cases{d(n),& for $x=n\in\nz$,\\
0,& else,\\}
\ee
and $K_\nu(z)$ and $Y_\nu(z)$ are Bessel functions (see e.g. \cite[p.\,66]{Magnus:1966}). Note, the series in \eref{11d} is not absolutely convergent due to the asymptotics \cite[p.\,139]{Magnus:1966}
\be
\label{rr3}
\hspace{-2cm}\eq{
K_\nu(z)&=\left(\frac{\pi}{2z}\right)^{\frac{1}{2}}\ue^{-z}\left(1+\Orr\left(\frac{1}{z}\right)\right),\quad-\frac{3\pi}{2}<\arg z<\frac{3\pi}{2},\\
Y_\nu(z)&=\left(\frac{1}{2}\pi z\right)^{-\frac{1}{2}}\sin\left(z-\frac{1}{2}\pi\nu-\frac{1}{4}\pi\right)\left(1+\Or\left(\frac{1}{z}\right)\right),\quad -\pi<\arg z<\pi. 
}
\ee
Owing to \eref{11d}, \eref{rr3}  and \eref{9a} we can estimate $N^{\mathrm{Osc.}}(x)$ as \cite{Wilton:1931a}
%
\be
\label{q12d}
\hspace{-2cm} N^{\mathrm{Osc.}}(x)=-\frac{x^{\frac{1}{4}}}{\pi\sqrt{2}}\sum\limits_{n=1}^{\infty}\frac{d(n)}{n^{\frac{3}{4}}}\cos\left(4\pi\sqrt{nx}-\frac{\pi}{4}\right)+\Or\left(x^{\epsilon}\right),\quad \epsilon>0 \quad x\rightarrow\infty,
\ee
and we recognize that only the first term on the r.h.s. meets the $\Omega$ result of Hardy \cite{Hardy:1917}. Relation \eref{q12d} suggests that $N^{\mathrm{Osc.}}(x)$ oscillates about zero with asymptotically vanishing mean value. This observation is corroborated by the asymptotics of its mean value \cite{Wilton:1931a}
\be
\label{q12e}
\overline{N^{\mathrm{Osc.}}}(x):=\frac{1}{x}\int\limits_{0}^{x}N^{\mathrm{Osc.}}(x')\ud x'=\Or\left(x^{-\frac{1}{4}}\right)\rightarrow 0, \quad x\rightarrow \infty. 
\ee
If we would not have taken into account the constant $\frac{1}{4}$ in \eref{11}, then $\overline{N^{\mathrm{Osc.}}}(x)$ would instead tend to $\frac{1}{4}$ for $x\rightarrow \infty$. Furthermore, \eref{11c} and \eref{11d} arise formally by applying the Vorono\"i summation formula \eref{13} to the Heaviside step function $f(x):=\Theta(x)$, which is, however, not an allowed test function in theorem \ref{12}. (Nevertheless, the formula \eref{13} still holds for the Heaviside step function by the calculations in \cite{Wilton:1932}). Therefore, we denote 
\be
\label{q12}
N^W(x):=x\ln x  +\left(2\gamma-1\right)x+\frac{1}{4},\quad x\in\rz_{>0},
\ee
as the ``Weyl term'' and $N^{\mathrm{Osc.}}(x)$ as the ``oscillatory term'' to $N(x)$.

It turns out that the Vorono\"i summation formula \eref{13} (see also \cite{Voronoi:1904}) can be interpreted as a trace formula for our noncompact quantum graph $\mathfrak{G}$ with pure Dirichlet boundary conditions specified in (\ref{2}) and (\ref{5}). In various articles (e.g. \cite{Rogosinski:1922,Oppenheim:1927,Landau:1927,Koshliakov:1928,Dixon:1931,Wilton:1932,Ferrar:1935,Ferrar:1937,Dixon:1937,Chandrasekharan:1968,Berndt:1972,Hejhal:1979,Endres:2010}) the authors investigate the Vorono\"i summation formula and specify proper function spaces for which the Vorono\"i summation formula is valid. In the following we employ the Vorono\"i summation formula as formulated in
\begin{theorem}[Wilton:\,1932, Dixon/Ferrar:\,1937, \cite{Wilton:1932,Dixon:1937}]
\label{12}
If, for any finite $k_0>0$, $f(k)$ is a real function of bounded variation in the interval $(0,k_0)$ and is continuous at $k=1,2,3,\ldots$, then
\begin{equation}
\label{13}
\hspace{-2cm}\eqalign{
\sum\limits_{n=1}^{\infty} & d(n) f(n)=\int\limits_0^{\infty}\left(\ln k+2\gamma\right)f(k)\ud k+\frac{f(0)}{4}\\
& \hspace{3.0cm}+2\pi\sum\limits_{n=1}^{\infty}d(n)\int\limits_0^{\infty}\left[\frac{2}{\pi}K_0\left(4\pi\sqrt{nk}\right)-Y_0\left(4\pi\sqrt{nk}\right)\right]f(k)\ud k, 
}
\end{equation}
provided that
\begin{itemize}
 \item $ \left(V_{0^+}^x f(k)\right)\ln x\rightarrow 0$ as $x\rightarrow 0^+$, 
 \item for some positive $\kappa$, $k^{\frac{1}{2}+\kappa}f(k)\rightarrow 0$ as $k\rightarrow\infty$,
 \item $f(k)$ is the (indefinite) integral of $f'(k)$ in $k\geq k_0$,
 \item for some positive $\kappa$ 
\begin{equation}
\label{13a}
 \int\limits_{\kappa}^{\infty}k^{\frac{1}{2}+\kappa}\left|f'(k)\right|\ud k<\infty.
\end{equation}
\end{itemize}
\end{theorem}
\hspace{-0.83cm}In (\ref{13}) $V_{0^+}^x f(k)$ denotes the total variation of $f(k)$ in $(0,x)$ (see e.g.\cite{Boyarsky1997}).

To interpret formula \eref{13} as a trace formula for our infinite quantum graph, we first notice that the l.h.s. is nothing else than the trace of the operator function $f\left((-\Delta)^{\frac{1}{2}}\right)$, i.e. with the wavenumbers $k_n:=\sqrt{\lambda_n}=n$
\be
\label{13aaq}
\sum\limits_{n=1}^{\infty}d(n)f(n)=\sum\limits_{n=1}^{\infty}d(n)f(k_n)=\tr f\left((-\Delta^{\frac{1}{2}})\right)
\ee
(under the conditions stated in theorem \ref{12}, $f\left((-\Delta)^{\frac{1}{2}}\right)$ is of trace class).

Furthermore, due to the imposed Dirichlet boundary conditions for the quantum graph, the length spectrum of the classical primitive periodic orbits of the graph $\mathfrak{G}$ is precisely given by the set of numbers $\left\{\lf_{p,n}:=\frac{2\pi}{n};\quad n=1,\ldots,\infty\right\}$ which allows us to interpret the series on the r.h.s. of \eref{13} as a series over the primitive periodic orbits by the replacement $n=\frac{2\pi}{\lfp}$ in the arguments of the Bessel functions $K_0$ and $Y_0$.
Under different conditions on the function $f$, as for instance in \cite{Hejhal:1979,Endres:2010}, the series on the r.h.s. in \eref{13} can be written as a double series
\begin{equation}
\label{100}
\hspace{-2.5cm}\eqalign{
& \sum\limits_{n=1}^{\infty}d(n)\int\limits_0^{\infty}\left[\frac{2}{\pi}K_0\left(2(2\pi)^{\frac{3}{2}}\sqrt{\frac{k}{\lf_{p,n}}}\right)-Y_0\left(2(2\pi)^{\frac{3}{2}}\sqrt{\frac{k}{\lf_{p,n}}}\right)\right]f(k)\ud k \\
& \hspace{3cm} =  \sum\limits_{l,m=1}^{\infty}\int\limits_0^{\infty}\left[\frac{2}{\pi}K_0\left(2(2\pi)^{\frac{3}{2}}\sqrt{\frac{mk}{\lf_{p,l}}}\right)-Y_0\left(2(2\pi)^{\frac{3}{2}}\sqrt{\frac{mk}{\lf_{p,l}}}\right)\right]f(k)\ud k.
}
\end{equation}
%
Therefore, the r.h.s. of \eref{13} has a purely classical interpretation. 

In the sequel we investigate the trace of the wave group (however, in ``euclidean'' time $t\rightarrow-\ui t$) and consider the trace of the heat kernel, the trace of the resolvent, some special zeta functions, a ''generalized gamma function'', a ``renormalized'' secular equation and a Selberg-like zeta function which is defined in terms of the classical periodic orbits and fulfills the analogue of the Riemann hypothesis.
\section{The trace of the wave group}
\label{14}
We shall need results from asymptotic analysis. We refer to \cite[p.\,320]{Zeidler:2006}  for the following theorem (in fact we think that there is a typographic mistake in \cite[p.\,320]{Zeidler:2006} concerning the sign of the second summand in equation (6.77) therein [cf. \hspace{-0.9mm}with \cite{Flajolet:1995}]).
\begin{theorem}[Zeidler, \cite{Zeidler:2006}]
\label{15}
Let $f:(0,\infty)\rightarrow\kz$ be a smooth function which together with all its derivatives is of 
rapid decay at infinity (i.e.\ $x^af^{(n)}(x)$ is bounded on $(b,\infty)$ for some $b>0$, for any $a\in\rz$ and for all $n\in\nz_0$) and possesses the asymptotic expansion
\begin{equation}
\label{16}
f(x)\sim\sum\limits_{n=0}^{\infty}b_{\lambda_n}x^{\lambda_n}, \quad x\rightarrow 0^+,
\end{equation}
with $-1=\lambda_0<\lambda_1<\lambda_2<\ldots$, then the function 
\begin{equation}
\label{17}
g(x):=\sum\limits_{n=1}^{\infty}f(nx)
\end{equation}
possesses the asymptotic expansion 
\begin{equation}
\label{18}
g(x)\sim\frac{1}{x}\left(\In_f^*-b_{-1}\ln x\right)+\sum\limits_{n=1}^{\infty}b_{\lambda_n}\zeta(-\lambda_n)x^{\lambda_n}, \quad x\rightarrow 0^+,
\end{equation}
with 
\begin{equation}
\label{18a}
\In_f^*:=\int\limits_{0}^{\infty}\left(f(x)-b_{-1}\frac{\ue^{-x}}{x}\right)\ud x.
\end{equation}
\end{theorem}
\hspace{-10mm} Here $\zeta(z)$ in (\ref{18}) denotes the Riemann zeta function.

Using for a certain Lambert series the identity \cite[p.\,467]{Knopp:1964}
\begin{equation}
\label{19}
\sum\limits_{n=1}^{\infty}\dz(n)q^n=\sum\limits_{n=1}^{\infty}\frac{q^n}{1-q^n}, \quad |q|<1,
\end{equation}
we obtain for the trace $\Th(t)$ of the (euclidean) wave group $\ue^{-t\sqrt{-\Delta}}$ where $-\Delta$ denotes the Dirichlet Laplacian, $-\Delta:=\left(-\Delta,\Dz(\Delta)\right)$ (see section \ref{1})
\begin{equation}
\label{20}
\eqalign{
\Th(t) & :=\tr\left(\ue^{-t\sqrt{-\Delta}}\right)=\sum\limits_{n=1}^{\infty}\dz(n)\ue^{-nt} \\
& =\sum\limits_{n=1}^{\infty}\frac{1}{\ue^{nt}-1}=\sum\limits_{n=1}^{\infty}f(nt), \quad t>0,
}
\end{equation}
where we have defined $f(t):=\frac{1}{\ue^{t}-1}$. It is obvious that $f(t)$ fulfills the requirements of theorem \ref{15}. Furthermore, for $f(t)$ we have the Laurent series (see e.g.\ \cite[p.\,25]{Magnus:1966}) 
\begin{equation}
\label{21}
f(t)=\sum\limits_{k=-1}^{\infty}\frac{B_{k+1}}{(k+1)!}t^{k},\quad 0<|t|<2\pi,
\end{equation}
where $B_k$ are the Bernoulli numbers ($B_0=1$, $B_1=-\frac{1}{2}$, $B_2=\frac{1}{6}$, $B_4=-\frac{1}{30},\ldots$ and $B_3=B_5=B_7
=\ldots=0$). For the calculation of 
\begin{equation}
\label{21a}
\In_f^*:=\int_{0}^{\infty}\left[\frac{1}{\ue^{t}-1}-\frac{\ue^{-t}}{t}\right]\ud t
\end{equation}
we use 
%
%
%
%
%
%
%
%
%
%
the following integral representation of the digamma function $\psi(z)=\frac{\Gamma'(z)}{\Gamma(z)}$ \cite[p.\,16]{Magnus:1966}
\be
\label{23a}
\psi(z)=\ln z+\int\limits_{0}^{\infty}\left[\frac{1}{t}-\frac{1}{1-\ue^{-t}}\right]\ue^{-zt}\ud t,\quad \re z>0,
\ee
from which we derive $\In_f^*=\gamma$ by evaluating \eref{23a} at $z=1$ and using $\psi(1)=-\gamma$. We then get by theorem \ref{15} and the identity $\zeta(-m)=-\frac{B_{m+1}}{m+1}$ for $m=1,2,3,\ldots$ \cite[p.\,19]{Magnus:1966} the asymptotic expansion
\begin{equation}
\label{24}
\Th(t)\sim-\frac{\ln t}{t}+\frac{\gamma}{t}+\frac{1}{4}+\sum\limits_{m=1}^{\infty}\frac{B_{m+1}^2}{(m+1)(m+1)!}(-t)^m, \quad t\rightarrow 0^+.
\end{equation}
The estimate $\Th(t)=\Or\left(\ue^{-t}\right)$ for $t\rightarrow\infty$ is trivial. 

Note that from the small-$t$ asymptotics (\ref{24}) one derives the leading asymptotic term for the counting function  (\ref{11}) using the Karamata-Tauberian theorem in the form \cite{Simon:1983}
\begin{equation}
\label{24a}
\lim_{t\rightarrow0^+}\left[-\frac{t^r}{\ln t} \tr\ue^{-t\sqrt{-\Delta}}\right]=c \quad \mbox{iff} \quad  \lim_{x\rightarrow\infty}\frac{N(x)}{x^r\ln x}=\frac{c}{\Gamma(r+1)}.
\end{equation}

We want to derive an exact expression for the ``error term'' in \eref{24} using a similar technique as in \cite{Hardy:1917}.
%

Let us define the spectral zeta function $\zeta_{\Delta}(s)$ of the Dirichlet Laplacian on the quantum graph $\mathfrak{G}$ (for a general definition see \cite{Singer:1973,Hawking:1977})
\be
\label{101}
\zeta_{\Delta}(s):=\tr(-\Delta)^{-s}=\sum\limits_{n=1}^{\infty}\frac{d(n)}{n^{2s}}=\sum\limits_{m,n=1}^{\infty}\frac{1}{(mn)^{2s}}=\zeta^2(2s),\quad \re s>\frac{1}{2},
\ee
respectively the spectral zeta function of $(-\Delta)^{\frac{1}{2}}$
\be
\label{101ab}
\zeta_{\Delta^{\frac{1}{2}}}(z):=\tr(-\Delta)^{-\frac{z}{2}}=\zeta^2(z),\quad \re z>1.
\ee
Using the inverse Mellin transform of the gamma function $\Gamma(z)$,
\be
\label{101a}
e^{-t} =\frac{1}{2\pi\ui}\int\limits_{\kappa-\ui\infty}^{\kappa+\ui\infty}t^{-z}\Gamma(z)\ud z, \quad \re t>0, \quad \kappa>0,
\ee
we then get from \eref{20}, \eref{101a} and \eref{101ab} (see the second footnote in \cite[p.\,5]{Hardy:1917})
\be
\label{102}
\eq{
\Th(t) &=\sum\limits_{n=1}^{\infty}\left[\frac{d(n)}{2\pi\ui}\int\limits_{\kappa-\ui\infty}^{\kappa+\ui\infty}(nt)^{-z}\Gamma(z)\ud z \right]\\
 &= \frac{1}{2\pi\ui}\int\limits_{\kappa-\ui\infty}^{\kappa+\ui\infty}t^{-z}\Gamma(z)\zeta^2(z) \ud z, \quad \re t>0, \quad \kappa>1.
}
\ee
Obviously, $\Theta_{\Delta^{\frac{1}{2}}}(t)$ is the inverse Mellin transform of $\Gamma(z)\zeta^2(z)$, and thus we obtain as a byproduct the following integral representation of the square of the Riemann zeta function 
\be
\label{102a}
\zeta^2(z)=\frac{1}{\Gamma(z)}\int\limits_{0}^{\infty}t^{z-1}\Theta_{\Delta^{\frac{1}{2}}}(t)\ud t,\quad \re z>1.
\ee
In order to derive an explicit expression for $\Theta_{\Delta^{\frac{1}{2}}}(t)$, we shift the contour of integration in \eref{102} to the left and use Cauchy's residuum theorem (see again the footnote in \cite[p.\,5]{Hardy:1917} in order to convince oneself that the steps are allowed). With the Laurent expansions \cite[p.\,2]{Magnus:1966} and \cite[p.\,807]{Abramowitz:1992} at $z=1$
\be
\label{103}
\eq{
t^{-z} &=\frac{1}{t}-\frac{\ln t}{t}(z-1)+\Or\left((z-1)^2\right),\\
\Gamma(z) &=1-\gamma(z-1)+\Or\left((z-1)^2\right),\\
\zeta^2(z) &=\frac{1}{(z-1)^2}+\frac{2\gamma}{(z-1)}+\gamma^2-2\gamma_1+\Or\left(z-1\right),
}
\ee
where $\gamma_1$ denotes the Stieltjes constant 
\be
\label{220}
\gamma_1=\underset{m\rightarrow\infty}{\lim}\left[\sum\limits_{l=1}^m\frac{\ln l}{l}-\frac{\ln^2m}{2}\right],
\ee
and the Laurent expansions at $z=-n$, $n\in\nz_0$,
\be
\label{241}
\eq{
t^{-z} &=t^n+\Or\left(z+n\right),\\
\Gamma(z) &=\frac{(-1)^n}{n!}\left[\frac{1}{z+n}+\psi(n+1)\right]+\Or\left(z+n\right),\\
\zeta^2(z) &=\frac{B_{n+1}^2}{(n+1)^2}+\Or\left(z+n\right) 
}
\ee
we obtain the exact expression ($N\in\nz$)
\be
\label{104}
\Th(t)=-\frac{\ln t}{t}+\frac{\gamma}{t}+\frac{1}{4}+\sum\limits_{n=1}^{N}\frac{B_{n+1}^2}{(n+1)(n+1)!}(-t)^n+J_{\kappa_{N}},
\ee 
where the error term $J_{\kappa_{N}}$ is given by 
\be
\label{106}
J_{\kappa_{N}}:=\frac{1}{2\pi\ui}\int\limits_{\kappa_{N}-\ui\infty}^{\kappa_{N}+\ui\infty}t^{-z}\Gamma(z)\zeta^2(z) \ud z, \quad -N<\kappa_{N}<-N-1.
\ee
It is obvious that
\be
\label{107}
J_{\kappa_{N}}=\orr\left(t^N\right),
\ee
and therefore \eref{24} is in fact an asymptotic expansion (see e.g. \cite[p.\,11]{Erdelyi:1956}). Since \cite[p.\,29]{Magnus:1966}
\be
\label{105}
|B_{2n}|>\frac{2(2n)!}{(2\pi)^{2n}},\quad B_{2n+1}=0, \quad n\in\nz,
\ee
we conclude that the radius of convergence of the formal power series in \eref{24} is equal to zero and thus we have proved the following theorem.
\begin{theorem}
\label{108}
$\Th(t)$, defined in $\eref{20}$, possesses the asymptotic expansion $\eref{24}$. Furthermore, it holds for the error term in the exact expression $\eref{104}$:
\be
\label{107a}
J_{\kappa_{N}}\not\rightarrow0 \quad \mbox{for} \quad \kappa_{N}\rightarrow-\infty,\quad \kappa_{N}\notin-\nz_0
\ee
in contrast to the case considered in $\cite{Hardy:1917})$.
\end{theorem} 

The function $f_t(k):=\ue^{-tk}$, $t>0$, fulfills the requirements of the trace formula, theorem \ref{12}, for our quantum graph and thus we get for the l.h.s.\ in (\ref{13}) exactly the trace of the wave group $\ue^{-t\sqrt{-\Delta}}$ (\ref{20}). Therefore, we evaluate the r.h.s.\ in (\ref{13}) and obtain 
\begin{equation}
\label{25}
\Th(t)=\Th^W(t)+\Th^\mathrm{po.}(t),
\end{equation}
where we have defined a ``Weyl term'' 
\begin{equation}
\label{26}
\eqalign{
\Th^W(t) & :=\int\limits_0^{\infty}(\ln k+2\gamma)f_t(k)\ud k+\frac{f_t(0)}{4} \\
           & =\int\limits_{0}^{\infty}\ue^{-tk}\ln k\ud k +2\gamma\int\limits_{0}^{\infty}\ue^{-tk}\ud k+\frac{1}{4}\\
           & =-\frac{\ln t}{t}+\frac{\gamma}{t}+ \frac{1}{4},\quad t>0, 
}
\end{equation}
and a ``periodic orbit term'' ($\lfp=\frac{2\pi}{n}$), $t>0$,
\begin{equation}
\label{27}
\fl\eqalign{
\Th^\mathrm{\po}(t) & :=2\pi\sum\limits_{n=1}^{\infty}d(n)\int\limits_0^{\infty}\left[\frac{2}{\pi}K_0\left(4\pi\sqrt{nk}\right)-Y_0\left(4\pi\sqrt{nk}\right)\right]f_t(k)\ud k\\
         &     =4\pi\sum\limits_{n=1}^{\infty}\dz(n)\int\limits_0^{\infty}\left[\frac{2}{\pi}K_0\left(4\pi\sqrt{n}y\right)-Y_0\left(4\pi\sqrt{n}y\right)\right]\ue^{-ty^2}y\ud y\\
         & =\frac{2}{t}\sum\limits_{n=1}^{\infty}\dz(n)\left[\exp\left(\frac{8\pi^3}{t\lfp}\right)\Ee\left(\frac{8\pi^3}{t\lfp}\right) -\exp\left(-\frac{8\pi^3}{t\lfp}\right)\Es\left(\frac{8\pi^3}{t\lfp}\right)\right].
}
\end{equation}
Here the integrals in the second line can be found in \cite[p.\,352]{Prudnikov:1986} respectively \cite[p.\,266]{Prudnikov:1986}, where $\Ee(x)$ and $\Es(x)$, $x>0$, are the exponential integral functions defined in \cite[p.\,342]{Magnus:1966}).
%
%
%
%
Formula \eref{27} has also been given in \cite{Oberhettinger:1972}. From their asymptotics for $x\rightarrow\infty$ \cite[pp.\,346, 347]{Magnus:1966} we obtain (after correcting a typographical error and replacing $\ue^{-x}$ by $\ue^x$ on p.347 l.c.)
\begin{equation}
\label{28}
e^x\Ee(x)-e^{-x}\Es(x)=-\sum\limits_{m=0}^{M}\frac{2(2m+1)!}{x^{2m+2}}+\Or\left(\frac{1}{x^{2M+4}}\right),\quad x\rightarrow \infty.
\end{equation}
Therefore, we get for the ``periodic orbit term'' (\ref{27}) for $t\rightarrow0^+$ the asymptotics ($x=\frac{4\pi^2n}{t}$)
\begin{equation}
\label{29}
\hspace{-1cm}\eq{
\Th^\mathrm{Osc.}(t) &=-\frac{4}{t}\sum\limits_{n=1}^{\infty}d(n)\left[\sum\limits_{m=0}^{M}\left(\frac{t}{4\pi^2}\right)^{2m+2}\frac{(2m+1)!}{n^{2m+2}}+\Or\left(\left(\frac{t}{n}\right)^{2M+4}\right)\right]\\
& =-\frac{4}{t}\sum\limits_{m=0}^{M}\left[(2m+1)!\left(\frac{t}{4\pi^2}\right)^{2m+2}\left[\sum\limits_{n=1}^{\infty}\frac{d(n)}{n^{2m+2}}\right ]\right]+\Or\left({t^{2M+3}}\right)\\
& =-\frac{4}{t}\sum\limits_{m=0}^{M}(2m+1)!\zeta^2(2m+2)\left(\frac{t}{4\pi^2}\right)^{2m+2}+\Or\left({t^{2M+3}}\right)\\
& =-\sum\limits_{m=0}^{M}\frac{B_{2m+2}^{2}}{(2m+2)^2(2m+1)!}t^{2m+1}+\Or\left({t^{2M+3}}\right),\quad t\rightarrow0^+,
}
\end{equation}
where we have used \eref{101} and $\zeta(2m+2)=\frac{(2\pi)^{2m+2}}{2(2m+2)!}\left|B_{2m+2}\right|$ for $m\in\nz_0$. Thus, we infer from (\ref{25}), (\ref{26}) and (\ref{29}) for $t\rightarrow0^+$
\begin{equation}
\label{30}
\Th(t)=-\frac{\ln t}{t}+\frac{\gamma}{t}+\frac{1}{4}-\sum\limits_{m=0}^{M}\frac{B_{2m+2}^{2}}{(2m+2)^2(2m+1)!}t^{2m+1}+\Or\left({t^{2M+3}}\right), 
\end{equation}
in complete agreement with (\ref{24}) resp. \eref{104}.

\section{The ``generalized spectral zeta'' function}
\label{30a}

We define the ``generalized spectral zeta function'' for $\sqrt{-\Delta}$ (motivated by \cite{Steiner:1987}) as follows
\begin{equation}
\label{31}
\eqalign{
\Zh(s,k) & :=\tr\left[\sqrt{-\Delta}+k\right]^{-s} \\
 & =\sum\limits_{n=1}^{\infty}\frac{\dz(n)}{(n+k)^s}, \quad \re s>1, \quad k\in\kz\setminus(-\nz), 
}
\end{equation}
(using for ${\frac{1}{(n+k)^s}:=\ue^{-s\ln(n+k)}}$ the principal branch of the logarithm) which is proportional to the Mellin transform of 
\begin{equation}
\label{32}
{\Th}(t,k):=\sum\limits_{n=1}^{\infty}\dz(n)\ue^{-(n+k)t}=\ue^{-kt}\Th(t),\quad t>0,\quad \re k>-1.
\end{equation}
We call ${\Th}(t,k)$ the trace of the ``generalized (euclidean) wave group''. The Mellin transform $\widehat{f}$ of a function $f$ is defined by 
\begin{equation}
\label{32aaaa}
\widehat{f}(s):=\int\limits_0^{\infty}t^{s-1}f(t)\ud t \quad 
\end{equation}
and defines a holomorphic function in a proper domain of definition. Due to the asymptotics \eref{24} of $\Th(t)$ for $t\rightarrow0^+$ and the trivial estimate ${\Th}(t,k)=\Or\left(\ue^{-(k+1)t}\right)$ for $t\rightarrow\infty$, the Mellin transform of $\Th(t)$ and therefore also of ${\Th}(t,k)$  exists at least for $\re s>1$ and $\re k>-1$. Similarly, as in \cite[p.\,310]{Zeidler:2006} summation and integration can be interchanged and we get
\begin{equation}
\label{32ab}
\eqalign{
{\Zh}(s,k) &=\frac{1}{\Gamma(s)}\widehat{{\Th}}(s,k)\\
         &=\frac{1}{\Gamma(s)}\int\limits_{0}^{\infty}t^{s-1}\ue^{-kt}\Th(t)\ud t,\quad \re s>1, \quad \re k>-1.
}
\end{equation}

In order to obtain the analytic continuation of ${\Zh}(s,k)$ as a function of $s$ (for fixed $\re k>-1$), 
we require the small-$t$ asymptotics of $\Th(t,k)$ for fixed $k$.
%
Using the Taylor expansion of $\ue^x$ and the asymptotic expansion \eref{24} we obtain for the asymptotics of ${\Th}(t,k)$ for fixed $k$ 
\be
\label{800}
\eq{
{\Th}(t,k) &= \ue^{-kt}{\Th}(t)\\
           &\sim\left[\sum\limits_{n=0}^{\infty}\frac{(-k)^nt^n}{n!}\right]\left[-\frac{\ln t}{t}+\frac{\gamma}{t}+\sum\limits_{m=0}^{\infty}\frac{B_{m+1}^2}{(m+1)(m+1)!}(-t)^m\right]\\
           &\sim-\frac{\ln t}{t}+\frac{\gamma}{t}+\sum\limits_{n=0}^{\infty}(-1)^n\frac{k^{n+1}}{(n+1)!}t^n \ln t\\
           &\phantom{\sim} +\sum\limits_{l=0}^{\infty}(-1)^l\left[a_l(k)-\frac{\gamma k^{l+1}}{(l+1)!}\right]t^l,\quad t\rightarrow0^+,\quad k\in\kz,
}
\ee
where we have defined
\be
\label{801}
a_l(k):=\sum\limits_{n+m=l}\frac{B^2_{m+1}k^n}{(m+1)(m+1)!n!}.
\ee
Therefore, we have proved the following theorem
\begin{theorem}
\label{1000}
${\Th}(t,k)$ possesses for fixed $k \in\kz$ and $t\rightarrow 0^+$ the asymptotics $\eref{800}$.
\end{theorem}
To proceed further, we need the following theorem \cite[p.\,307]{Zeidler:2006}, \cite{Flajolet:1995}
\begin{theorem}[\cite{Zeidler:2006}, \cite{Flajolet:1995}]
\label{32a}
Let $f(t):(0,\infty)\rightarrow\kz$ be a continuous function and of rapid decay at infinity (i.e.\ $t^af(t)$ is bounded on $(b,\infty)$ for some $b>0$ and for any $a\in\rz$). Furthermore, there exists an asymptotic expansion for $f(t)$, $t\rightarrow0^+$, of the form
\begin{equation}
\label{32b}
f(t)\sim\sum\limits_{j=1}^{\infty}a_jt^{\alpha_j}(\ln t)^{m_j}, \quad t\rightarrow 0^+,
\end{equation}
where $\alpha_j$ is an increasing sequence of real numbers tending to $\infty$ (maybe finitely many are negative) and $m_j\in\nz_{0}$ are arbitrary. Then the Mellin transform $\widehat{f}(s)$ has a meromorphic extension to all of $\kz$ with poles at $s=-\alpha_j$ and respective principal part $\bigl(\Pp[\widehat{f}]_{-\alpha_j}\bigr)(s)$ (the sum gathers all contributions of $\eref{32b}$ corresponding to the same value of $\alpha_j$)
\begin{equation} 
\label{32c}
\bigl(\Pp[\widehat{f}]_{-\alpha_j}\bigr)(s)=\sum\limits_{m_j}a_j(-1)^{m_j}\frac{(m_j)!}{\left({\alpha_j+s}\right)^{m_j+1}}.
\end{equation}
\end{theorem}
It is obvious that ${\Th}(t,k)$ fulfills as a function of $t$ the requirements of theorem \ref{32a} for $\re k>-1$. The reciprocal gamma function $\Gamma(s)^{-1}$ possesses at $s=1$ the Taylor expansion \cite[p.\,2]{Magnus:1966}
\be
\label{802}
\Gamma(s)^{-1}=1+\gamma(s-1)+\Or\left((s-1)^2\right),
\ee
and at $s=-n$, $n\in\nz_0$, the Taylor expansion \cite[p.\,2]{Magnus:1966}
\be
\label{803}
\Gamma^{-1}(s)=(-1)^nn!(s+n)+\Or\left((s+n)^2\right),\quad n\in\nz_0.
\ee
Combining theorem \ref{32a}, \eref{32}, \eref{800}, \eref{802} and \eref{803} we obtain the following theorem
\begin{theorem}
\label{804}
$\Zh(s,k)$, defined in $\eref{31}$, $\re k>-1$ fixed, has as a function of $s$ a meromorphic extension to all of $\kz$ with poles at $s=-(n-1)$, $n\in\nz_0$, and respective principal parts
\be
\label{805}
\eq{
\Pp\left[\Zh\right]_{s=1}(s,k) &=\frac{1}{(s-1)^2}+\frac{2\gamma}{s-1}\\
\Pp\left[\Zh\right]_{s=-n}(s,k) &=-\frac{k^{n+1}}{n+1}\frac{1}{s+n}, \quad n\in\nz_0.
}
\ee
\end{theorem}
Note that the pole at $s=1$ is independent of $k$ in contrast to the poles at $s=-n$, $n\in\nz_0$, which are absent for $k=0$. Furthermore, \eref{805} is in complete agreement with the obvious identity (see \eref{101} and \eref{31})
\be
\label{805a}
\Zh(s,0)=\zeta^2(s).
\ee
Finally, let us investigate the holomorphy of ${\Zh}(s,k)$ defined in (\ref{31}) as a function on $\kz^2$ in $s$ and $k$. For this reason, we need in the sequel Hartogs' theorem \cite{Hartogs:1906} (this version is stated and improved in \cite{Filippov:1973}).
\begin{theorem}[Hartogs' theorem]
\label{55a}
Let $K$ be a domain in the space of the complex variables $k=\left(k_1,\ldots,k_n\right)$, and let $S_0\subset S$ be domains in the space of the complex variables $s=\left(s_1,\ldots,s_m\right)$. If $f(s,k)$ is analytic with respect to the set of variables in the domain $s\in S_0$, $k\in K$, and if for each fixed $k\in K$ it is analytic with respect to $s$ in the domain $S$, then $f(s,k)$ is analytic with respect to the set of variables $s$, $k$ in the domain $S\times K$.
\end{theorem} 
Due to the normal convergence of the series in \eref{31}, the function $\Zh(s,k)$, $\re s>1$ fixed, is as a function of $k$ holomorphic on $\kz\setminus\left\{z\leq-1\right\}$. Combining this with theorem \ref{804} we infer the following theorem (see \cite[p.\,135]{Cartan:1995}):
\begin{theorem}
\label{806}
$\Zh(s,k)$, defined in $\eref{31}$, is an analytic function on
\be
\label{807}
\mathcal{S}_0\times\widetilde{\mathcal{K}}:=\left\{(s,k)\in\kz\times\kz; \quad  \re s>1, \quad  k\in\kz\setminus\left\{z\leq-1\right\}\right\}.
\ee
\end{theorem}
Now, we apply theorem \ref{55a} and infer from the above results the following corollary ($\mathcal{K}\subset\widetilde{\mathcal{K}}$):
\begin{cor}
\label{808}
$\Zh(s,k)$ is an analytic function on  
\be
\label{809}
\mathcal{S}\times\mathcal{K}:=\left\{(s,k)\in\kz\times\kz; \quad s\neq-(n-1), \quad n\in\nz_0, \quad \re k>-1\right\}.
\ee
\end{cor}
\section{The ``regularized trace'' of the resolvent of $(-\Delta)^{\frac{1}{2}}$}
\label{200}
Since the trace of the resolvent of $(-\Delta)^{\frac{1}{2}}$ does not exist ($\Zh(s,k)$, defined in \eref{31}, possesses a pole at $s=1$ for all $\re k>-1$ [see theorem \ref{804})]), we introduce the ``regularized trace'' of the resolvent ($\tau(0)=0$)
\be
\label{201}
\tau(k):=\sum\limits_{n=1}^{\infty}d(n)\left[\frac{1}{n+k}-\frac{1}{n}\right],\quad k\notin-\nz,
\ee
which is, due to \eref{9a}, absolutely convergent on $\rz\setminus\left\{-\nz\right\}$. 
It is not difficult to see that the imaginary part of $\tau(k)$ determines the counting function $N(x)$ of the wave numbers of our quantum graph, defined in \eref{11}, via the formula ($x>0$)
\be
\label{201aaz}
N(x):=\lim_{\epsilon\rightarrow0}\frac{1}{\pi}\int\limits_{0}^{x}\im\tau(-k-\ui\epsilon)\ud k.
\ee
We can write $\tau(k)$ as a double series, similarly as in \eref{100}, and with the relation \cite[p.\,13]{Magnus:1966}
\be
\label{202}
\psi(z)=-\gamma+\sum\limits_{n=0}^{\infty}\left[\frac{1}{n+1}-\frac{1}{z+n}\right], \quad z\notin-\nz_0,
\ee
we get
\be
\label{203}
\eq{
\tau(k) &=\sum\limits_{n=1}^{\infty}\frac{1}{n}\left[\sum\limits_{m=1}^{\infty}\left[\frac{1}{m+\frac{k}{n}}-\frac{1}{m}\right]\right]\\
        &=-\sum\limits_{n=1}^{\infty}\frac{1}{n}\left[\psi\left(1+\frac{k}{n}\right)+\gamma\right],\quad k\notin-\nz.
}
\ee

In order to determine the asymptotics of $\tau(k)$ for $k\rightarrow\infty$ and for further applications we use a theorem of \cite{Flajolet:1995} (see the definitions therein for the notations used below).  To apply the theorem assume that a harmonic sum of the form
\be
\label{211}
G(x):=\sum\limits_{n=1}^{\infty}\lambda_ng(\mu_nx)
\ee
is given and the Mellin transform $\hat{g}(s)$ of the ``base function'' $g$ possesses the fundamental strip (see definition 1 in \cite{Flajolet:1995})
\be
\label{200a}
<\alpha,\beta>:=\left\{z\in \kz; \ \alpha<\re z<\beta\right\}.
\ee
Let us furthermore assume that the Dirichlet series
\be
\label{212}
\Lambda(s):=\sum\limits_{n=1}^{\infty}\lambda_n\mu_n^{-s}
\ee
possesses a half-plane of simple (not necessarily absolute) convergence $\re s>\sigma_c$. Assume further that
\begin{itemize}
	\item the half-plane of convergence of $\Lambda(s)$ intersects the fundamental strip of $\hat{g}(s)$, and let $\alpha':=\max(\alpha,\sigma_c)$,
	\item the functions $\hat{g}(s)$ and $\Lambda(s)$ admit a meromorphic continuation into the strip $<\eta,\beta>$, $\eta<\alpha$ (respectively $<\alpha,\eta>$, $\eta>\beta$) and are analytic on $\re s=\eta$,
	\item on the closed strip $\eta\leq\re s\leq\frac{\alpha'+\beta}{2}$ (respectively $\frac{\alpha'+\beta}{2}\leq\re s\leq\eta$) the function $\hat{g}(s)$ is of  fast decrease ($\hat{g}(s)=\Or\left(|s|^{-r}\right)$, $|s|\rightarrow\infty$, $r>0$ arbitrary [definition 4 in \cite{Flajolet:1995}]) and the function $\Lambda(s)$ is of slow increase ($\Lambda(s)=\Or\left(|s|^r\right)$, $|s|\rightarrow \infty$, for some $r>0$ [definition 4 in \cite{Flajolet:1995}]).   
\end{itemize}
Then the theorem of \cite{Flajolet:1995} states:
\begin{theorem}[Flajolet/Gourdon/Dumas:\,1995, \cite{Flajolet:1995}]
\label{213}
Let the above conditions for $g$, $\hat{g}$ and $\Lambda$ be fulfilled. Then the following asymptotic formula holds:
\be
\label{214}
G(x)\sim _{(-)}^{ \phantom{(} + \phantom{)}}\sum_{P}\res\left[\hat{g}(s)\Lambda(s)x^{-s}\right]+\Or\left(x^{-\eta}\right), \quad x\rightarrow0, \ (\infty),
\ee
where $G(x)$ is defined in $\eref{211}$ and the sum on the r.h.s in $\eref{214}$ is over the set $P$ of poles in $<\eta,\alpha'>$ (respectively $<\alpha',\eta>$).
\end{theorem}
We write \eref{201} in the form
\be
\label{235}
\tau(k)=-\sum\limits_{n=1}^{\infty}\frac{d(n)}{n}\frac{\frac{k}{n}}{1+\frac{k}{n}},\quad k>0,
\ee
and consider \eref{235} as a harmonic sum as defined in \eref{211}. Comparing \eref{235} with \eref{211} we identify
\be
\label{236}
g(x):=-\frac{x}{x+1}, \quad \mu_n=\frac{1}{n} \quad \mbox{and}\quad \lambda_n=\frac{d(n)}{n}
\ee
and obtain (see \eref{101} and \cite[p.\,307,308]{Magnus:1954})
\be
\label{237}
\eq{
\Lambda(s) &=\sum\limits_{n=1}^{\infty}\frac{d(n)}{n^{1-s}}=\zeta^2(1-s),\quad  s\neq0,\\
\hat{g}(s) &=\frac{\pi}{\sin(s\pi)}, \quad s\notin\gz.
}
\ee
We realize that $\hat{g}$ possesses the fundamental strip $<-1,0>$ (the restrictions on the domain of definitions stated in \eref{237} pertain to the analytic continuation of the corresponding functions). Furthermore, $\Lambda(s)$ is of slow increase and $\hat{g}(s)$ is of fast decrease (see e.g. \cite[p.\,314]{Magnus:1954} and \cite[p.\,74]{Abramowitz:1992}). However, $\Lambda$ admits the half-plane of convergence $\re s<0$. Therefore, we can't directly apply theorem \ref{213} since a half-plane of convergence to the right is required therein. But the theorem still holds true in our case. We omit the proof of this assertion since the proof is a simple modification of the proof in \cite{Flajolet:1995}. The Laurent expansions of $\Lambda$ and $\hat{g}$ at $s=0$ are given by \eref{103} (replacing $z$ by $1-s$) and \cite[p.\,75]{Abramowitz:1992}
\be
\label{238}
\hat{g}(s)=\frac{1}{s}+\frac{\pi^2s}{6}+\Or\left(s^3\right).
\ee
Furthermore, the Laurent expansion  of $\hat{g}(s)$ at $s=1$ is given by (see \eref{103})
\be
\label{238a}
\hat{g}(s)=-\frac{1}{s-1}+\Orr((s-1)).
\ee
Thus, the Laurent expansion of $\Lambda\hat{g}$ at $s=0$ is given by
\be
\label{239}
\Lambda(s)\hat{g}(s)=\frac{1}{s^3}-\frac{2\gamma}{s^2}+\frac{\frac{\pi^2}{6}+\gamma^2-2\gamma_1}{s}+\Or(1)
\ee
and at $s=1$ by ($\zeta(0)=-\frac{1}{2}$ \cite[p.\,19]{Magnus:1966})
\be
\label{239a}
\Lambda(s)\hat{g}(s)=-\frac{1}{4(s-1)}+\Orr(1).
\ee
Since the next relevant pole of $\Lambda(s)\hat{g}(s)$ is located at $s=2$ we therefore obtain by theorem \ref{213} ($k\rightarrow\infty$)
\be
\label{240}
\tau(k)\sim-\frac{1}{2}\ln^2 k-2\gamma\ln k+D+\frac{1}{4k}+\Orr\left(\frac{1}{k^2}\right)
\ee
with the constant $D$ given by 
\be
\label{241a}
D=2\gamma_1-\frac{\pi^2}{6}-\gamma^2.
\ee

It is easy to derive the asymptotics of $\tau(k)$ for $k\rightarrow 0$. For this, we write the summands in \eref{235} as a geometric series and obtain (see \eref{101})
\be
\label{400}
\eq{
\tau(k) &=\sum\limits_{n=1}^{\infty}\frac{d(n)}{n}\left[\sum\limits_{m=1}^{\infty}\left(-\frac{k}{n}\right)^m\right]\\
        &=\sum\limits_{m=1}^{\infty}(-k)^m\left[\sum\limits_{n=1}^{\infty}\frac{d(n)}{n^{m+1}}\right]=\sum\limits_{m=1}^{\infty}\zeta^2(m+1)(-k)^m, \quad |k|<1,
}
\ee
where in the last line the interchange of the order of summations is allowed since the double series is absolutely convergent. Combining the above results we have proved the following theorem. 
\begin{theorem}
\label{233}
The regularized trace of the resolvent $\tau(k)$, defined in $\eref{201}$, possesses the representation $\eref{203}$ and the asymptotics $\eref{240}$ for $k\rightarrow\infty$ and the absolutely convergent series expansion $\eref{400}$ for $|k|<1$.
\end{theorem}
\hspace{-0.81cm}It is worthwhile to remark that the logarithmically divergent terms in \eref{240} are nothing else than the leading Weyl contributions to $\tau(k)$. Indeed, if one calculates the analytic continuation of the imaginary part of \eref{240} (using $\lim_{\epsilon\rightarrow0}\ln(-k-\ui \epsilon)=\ln k-\ui\pi$, $k>0$) one obtains from \eref{201aaz} precisely the two leading terms of the Weyl term in $N^W(x)$, see eq. \eref{q12}.
\section{A ``regularized spectral zeta function'' for $(-\Delta)^{\frac{1}{2}}$ and a ``generalized gamma function''}
\label{600}
We introduce the ``regularized spectral zeta function'' (see \eref{101} and \eref{31} and theorem \ref{806}) 
\be
\label{601}
\eq{
\tZh(s,k) &:=\sum\limits_{n=1}^{\infty}d(n)\left[\frac{1}{(n+k)^s}-\frac{1}{n^s}\right]\\
            &=\Zh(s,k)-\zeta^2(s), \quad \re s>0, \quad k\in\widetilde{\mathcal{K}},
}
\ee
which is holomorphic at $s=1$ for all $k\in\widetilde{\mathcal{K}}=\kz\setminus\left\{z\leq-1\right\}$ in contrast to $\Zh(s,k)$. We immediately notice the relations (see \eref{201})
\be
\label{602}
\tZh(s,0)=0,\quad s\in\kz , \quad \tZh(1,k)=\tau(k),\quad k\in\widetilde{\mathcal{K}}.
\ee
Because of theorem \ref{804}, corollary \ref{808}, \eref{103} and \eref{601} we infer the following corollary.
\begin{cor}
\label{603}
$\tZh(s,k)$ possesses an analytic continuation on 
\be
\label{604}
\widetilde{\mathcal{S}}\times\mathcal{K}:=\left\{(s,k)\in\kz\times\kz;  \quad s\neq-n, \quad n\in\nz_0,\quad \re k>-1\right\}
\ee
with respective principal parts $\eref{805}$ at $s=-n$, $n\in\nz_0$.
\end{cor}
Due to \eref{103} we have
\be
\label{605a}
\Fp\left[\zeta^2(s)\right]_{s=1}:=\underset{s\rightarrow 1}{\lim}\left[\zeta^2(s)-\frac{1}{(s-1)^2}-\frac{2\gamma}{s-1}\right]=\widetilde{\gamma},
\ee
where we have introduced the finite part $\Fp[f(s)]_{s=s_0}$ of a function $f(s)$ at $s=s_0$ (see e.g. \cite[p.\,55]{Zeilder:2009}) and the ``generalized Euler constant'' (see \eref{103})
\be
\label{605b}
\widetilde{\gamma}:=\gamma^2-2\gamma_1.
\ee
With  
\be
\label{605c}
\Fp\left[\tZh(s,k)\right]_{s=1}=\tZh(1,k)=\tau(k),\quad \re k>-1,
\ee
we obtain from \eref{805}, \eref{201} and \eref{601} 
\be
\eq{
\Fp\left[\Zh(s,k)\right]_{s=1} &:=\underset{s\rightarrow1}{\lim}\left[\Zh(s,k)-\frac{1}{(s-1)^2}-\frac{2\gamma}{s-1}\right]\\
 &=\widetilde{\gamma}+\tau(k)\\
 &=\widetilde{\gamma}+\sum\limits_{n=1}^{\infty}d(n)\left[\frac{1}{n+k}-\frac{1}{n}\right],\quad \re k>-1,
}
\ee
which should be compared with a very similar relation for the Hurwitz zeta function \cite[p.\,22,23]{Magnus:1966} 
\be
\label{607}
\eq{
\Fp\left[\zeta(s,k)\right]_{s=1} &:=\underset{s\rightarrow1}{\lim}\left[\zeta(s,k)-\frac{1}{s-1}\right]  =\gamma+\sum\limits_{n=0}^{\infty}\left[\frac{1}{n+k}-\frac{1}{n+1}\right]\\
                                                                    &=-\psi(k)=-\frac{\Gamma'(k)}{\Gamma(k)},\quad k\notin-\nz_0.
}
\ee
This depicts an analogy between $\Zh(s,k)$ and $\zeta(s,k+1)$ and thus we search for a ``generalized gamma function'' $\widetilde{\Gamma}(k)$ satisfying the corresponding relation to \eref{607}. We define in analogy to \cite[p.\,1]{Magnus:1966}
\be
\label{608a}
\frac{1}{\Gamma(k)}:=k\ue^{\gamma k}\prod\limits_{n=1}^{\infty}\left[\left(1+\frac{k}{n}\right)\ue^{-\frac{k}{n}}\right],\quad k\in\kz,
\ee
the ''generalized gamma function'' by a Weierstraß product ($\widetilde{\Gamma}(0)=1)$):
\be
\label{608}
\frac{1}{\widetilde{\Gamma}(k)}:=\ue^{\widetilde{\gamma}k}\prod\limits_{n=1}^{\infty}\left[\left(1+\frac{k}{n}\right)\ue^{-\frac{k}{n}}\right]^{d(n)},\quad k\in\kz.
\ee
The relation
\be
\label{609}
\frac{1}{\widetilde{\Gamma}(n)}=0,\quad n\in-\nz,
\ee
is obvious, which shows that the zeros of $\widetilde{\Gamma}(s)^{-1}$ are precisely located at the negative wavenumbers $-k_n=-n$ of the quantum graph $\mathfrak{G}$ and the orders of these zeros are exactly given by the associated multiplicities $d(n)$. In order to prove the normal convergence of the Weierstraß product \eref{608}, it is sufficient to prove the normal convergence of the series (cp. \cite[p.\,34]{Whittaker:1996})
\be
\label{610}
\sum\limits_{n=1}^{\infty}d(n)\left[\left(1+\frac{k}{n}\right)\ue^{-\frac{k}{n}}-1\right],\quad k\in\kz.
\ee
This is shown by the simple calculation
\be
\label{611}
\eq{
\left[1+\frac{k}{n}\right]\ue^{-\frac{k}{n}}-1 &=\left[1+\frac{k}{n}\right]\left[1-\frac{k}{n}+\Or\left(\left(\frac{k}{n}\right)^2\right)\right]-1\\
& =\Or\left(\left(\frac{k}{n}\right)^2\right), \quad k\in\kz,
}
\ee
and \eref{9a}. Furthermore, due to the normal convergence of $\widetilde{\Gamma}(k)^{-1}$ on $\kz$ we obtain for the ``generalized digamma function'' 
\be
\label{612}
\widetilde{\psi}(k):=\frac{\widetilde{\Gamma}'(k)}{\widetilde{\Gamma}(k)}=-\widetilde{\gamma}-\sum\limits_{n=1}^{\infty}d(n)\left[\frac{1}{n+k}-\frac{1}{n}\right],\quad k\notin-\nz,
\ee
which should be compared with the relation \eref{202} for the ``ordinary'' digamma function $\psi(z)$. Additionally, due to the normal convergence of $\widetilde{\Gamma}(k)^{-1}$ one can rearrange the product in \eref{608} and obtains using the identity \eref{608a}
\be
\label{613}
\eq{
\frac{1}{\widetilde{\Gamma}(k)} &=\ue^{\widetilde{\gamma}k}\prod\limits_{n=1}^{\infty}\left\{\ue^{-\frac{\gamma k}{n}}\left[\ue^{\frac{\gamma k}{n}}\prod\limits_{m=1}^{\infty}\left(1+\frac{k}{n}\frac{1}{m}\right)\ue^{-\frac{k}{n}\frac{1}{m}}\right]\right\}\\
 &=\ue^{\widetilde{\gamma}k}\prod\limits_{n=1}^{\infty}\frac{\ue^{-\frac{\gamma k}{n}}}{\left(\frac{k}{n}\right)\Gamma\left(\frac{k}{n}\right)}=\ue^{\widetilde{\gamma}k}\prod\limits_{n=1}^{\infty}\frac{\ue^{-\frac{\gamma k}{n}}}{\Gamma\left(1+\frac{k}{n}\right)},\quad k\in\kz,
}
\ee
which yields
\be
\label{613a}
\widetilde{\psi}(k)=-\widetilde{\gamma}+\sum\limits_{n=1}^{\infty}\frac{1}{n}\left[\psi\left(1+\frac{k}{n}\right)+\gamma\right],\quad k\notin-\nz.
\ee
By \eref{203} we infer the analogous relation to \eref{607}
\be
\label{613b}
\Fp\left[\Zh(s,k)\right]_{s=1}=-\widetilde{\psi}(k)=-\frac{\widetilde{\Gamma}'(k)}{\widetilde{\Gamma}(k)}.
\ee
We have proved the following theorem
\begin{theorem}
\label{614}
\begin{enumerate}
 \item[i)]  The generalized gamma function $\widetilde{\Gamma}(k)$, defined in $\eref{608}$, is a meromorphic function on $\kz$ whose poles are exactly located at the negative wavenumbers $-k_n$, $n\in\nz$, of the quantum graph $\mathfrak{G}$ and the orders of these poles are given by the associated multiplicities $d(n)$ of the eigenvalues $k_n^2$. Furthermore, the generalized gamma function $\widetilde{\Gamma}(k)$ fulfills the relations $\eref{612}$, $\eref{613}$ and $\eref{613b}$.
 \item[ii)] The generalized digamma function $\widetilde{\psi}(k)$, defined in $\eref{612}$, is a meromorphic function on $\kz$ whose poles are exactly located at the negative wavenumbers $-k_n=-n$, $n\in\nz$, of the quantum graph $\mathfrak{G}$. The residues of the poles are given by the associated multiplicities $d(n)$. Furthermore, 
the generalized digamma function $\widetilde{\psi}(k)$ fulfills the relations $\eref{612}$, $\eref{613a}$ and $\eref{613b}$.
\end{enumerate}
\end{theorem}
Finally, comparing \eref{612} with \eref{201} we obtain 
\begin{cor}
\label{614a}
The regularized trace of the resolvent of $(-\Delta)^{\frac{1}{2}}$ on the quantum graph $\mathfrak{G}$, defined in $\eref{201}$, has the compact representation 
\be
\label{614b}
\tau(k)=-\widetilde{\gamma}-\widetilde{\psi}(k).
\ee
\end{cor}
\section{A Stirling-like formula}
\label{a1}
We want to derive a Stirling-like formula for $\widetilde{\Gamma}(k)$, defined in \eref{608}. If we define 
\be
\label{a2}
\widetilde{Z}(k):=\frac{1}{\widetilde{\Gamma}(k)},\quad k\in\kz,
\ee
we obtain (see \eref{608})
\be
\label{a3}
\widetilde{Z}(0)=1.
\ee
Furthermore, it holds by \eref{612} and \eref{614b}
\be
\label{a4}
\dk \ln\widetilde{Z}(k)=-\frac{\widetilde{\Gamma'}(k)}{\widetilde{\Gamma}(k)}=-\widetilde{\psi}(k)=\widetilde{\gamma}+\tau(k),\quad k\notin-\nz.
\ee 
Since $\widetilde{Z}(k)$ is positive for nonnegative $k$ and by \eref{a3} and \eref{a4} we obtain
\be
\label{a5}
\eq{
\int\limits_{0}^{k}\frac{\ud\phantom{k'}}{\ud k'}\ln\widetilde{Z}(k')\ud k' &=\ln\widetilde{Z}(k)-\ln\widetilde{Z}(0)=\ln\widetilde{Z}(k)\\
                                                  &=\widetilde{\gamma}k+\int\limits_{0}^{k}\tau(k')\ud k',\quad k\in\rz_{\geq0}.
}
\ee
Therefore, we obtain (see \eref{605b})
\be
\label{a6}
\widetilde{\Gamma}(k)=\exp\left(-\widetilde{\gamma}k-\int\limits_{0}^{k}\tau(k')\ud k'\right),\quad k\in\kz.
\ee
If we split the integral in \eref{a6} in two parts
\be
\label{a7}
\int\limits_{0}^{k}\tau(k')\ud k'=\int\limits_{0}^{K}\tau(k')\ud k'+\int\limits_{K}^{k}\tau(k')\ud k',\quad 1<<K<<k
\ee
and insert in the second integral in \eref{a7} the asymptotic expansion \eref{240} we obtain
\be
\label{a8}
\hspace{-1cm}\eq{
\int\limits_{K}^{k}\tau(k')\ud k' &=-\int_{K}^{k}\frac{\ln^2 k'}{2}\ud k'-2\gamma\int_{K}^{k}\ln k'\ud k'+D\int_{K}^{k}\ud k'+\int_{K}^{k}\frac{\ud k'}{4k'}\\
 &\phantom{=}+ \int_{K}^{k}r(k')\ud k', \quad r(k')=\Orr\left(\frac{1}{k'^2}\right), \quad k'\rightarrow\infty. 
}
\ee
Using \cite[p.\,124]{Prudnikov:1986b} 
\be
\label{a10}
\fl\int\frac{\ln x}{x}\ud x=\frac{\ln^2 x}{2}, \ \int \ln x\ud x=x\ln x-x,\ \int \ln^2 x\ud x=x\ln^2x-2x\ln x+2x,
\ee
we obtain by \eref{a6} and \eref{a8} 
\begin{theorem}
\label{aa10}
The generalized gamma function $\widetilde{\Gamma}(k)$ defined by the Weierstraß product $\eref{608}$ satisfies asymptotically for $k\rightarrow\infty$ the Stirling-like formula
\be
\label{a11}
\hspace{-1.5cm}\widetilde{\Gamma}(k)=\exp\left(\frac{k\ln^2 k}{2}+(2\gamma-1)k\ln k+\left(1+\frac{\pi^2}{6}-2\gamma\right)k-\frac{\ln k}{4}+\Orr(1)\right).
\ee
\end{theorem}
\hspace{-0.85cm}From \eref{a4} and theorem \ref{233} we obtain the following 
\begin{cor}
\label{b1r}
The generalized digamma function $\widetilde{\psi}(k)$, defined in $\eref{612}$, possesses the following absolutely convergent series expansion for $|k|<1$, respectively for $k\rightarrow\infty$ the asymptotic behaviour
\be
\eq{
\widetilde{\psi}(k) &=2\gamma_1-\gamma^2+\sum\limits_{m=1}^{\infty}\zeta^2(m+1)(-1)^{m+1}k^m,\quad |k|<1\\
                    &=\frac{1}{2}\ln^2k+2\gamma\ln k +\frac{\pi^2}{6}-\frac{1}{4k}+\Or\left(\frac{1}{k^2}\right), \quad k\rightarrow\infty. 
}
\ee
\end{cor}
\section{The trace of the heat kernel and the spectral zeta function of $-\Delta$}
\label{32aaaaa}
For the trace ${\T}(t,k)$ of the ``generalized heat kernel'' $\ue^{-k^2t}\ue^{t\Delta}$ we get
\begin{equation}
\label{33}
\eqalign{
{\T}(t,k) &:=\tr\ue^{t\left(\Delta-k^2\right)}=:\ue^{-k^2t}\Theta_{\Delta}(t)=\sum\limits_{n=1}^{\infty}\dz(n)\ue^{-\left(n^2+k^2\right)t} \\
        &=\ue^{-k^2t}\sum\limits_{n=1}^{\infty}\left(\sum\limits_{m=1}^{\infty}\ue^{-n^2m^2t}\right)\\
        &=\ue^{-k^2t}\sum\limits_{n=1}^{\infty}\omega\left(\frac{n^2t}{\pi}\right), \quad t>0, \quad k\in\kz,
}
\end{equation}
with
\be
\label{33a}
\omega(x):=\sum\limits_{n=1}^{\infty}\ue^{-\pi n^2x}=\frac{1}{2}\left[\theta_3\left(0,\frac{\ui x}{\pi}\right)-1\right],
\ee
where $\theta_3\left(z,\tau\right)$ is the elliptic theta function  \cite[p.\,371]{Magnus:1966}. From
\begin{equation}
\label{34}
\theta_3(0,\ui x)=
\cases{
\frac{1}{\sqrt{x}} +\Or\left(\frac{\ue^{-\frac{\pi}{x}}}{\sqrt{x}}\right), & $x\rightarrow 0^+$, \\
1+\Or\left(\ue^{-\pi x}\right), & $x\rightarrow \infty$,\\ 
}
\end{equation}
we obtain the asymptotics
\be
\label{34a}
\omega\left(x^2\right)=
\cases{
\frac{1}{2x}-\frac{1}{2} +\Or\left(\frac{\ue^{-\frac{\pi}{x^2}}}{x}\right), & $x\rightarrow 0^+$, \\
\Or\left(\ue^{-\pi x^2}\right), & $x\rightarrow \infty$,\\ 
}
\ee
and thus we can apply theorem \ref{15} identifying $f(x):=\omega\left(x^2\right)$ and subsequently setting $x=\sqrt\frac{t}{\pi}$.
We then obtain from (\ref{33}) and  (\ref{34}) ($k\in\kz$ fixed, and any $\eta>0$)
\begin{equation}
\label{35}
{\T}(t,k) \sim \ue^{-k^2t}\left[\sqrt{\frac{\pi}{t}}\left(\In_f^{**}-\frac{1}{4}\ln\left(\frac{t}{\pi}\right)\right)+\frac{1}{4}+\Orr\left(t^{\eta}\right)\right], \quad t\rightarrow 0^+,
\end{equation}
where we have defined
\begin{equation}
\label{35a}
\In_f^{**}:=\int\limits_{0}^{\infty}\left[\omega\left(x^2\right)-\frac{\ue^{-x}}{2x}\right]\ud x.
\end{equation}
Using the integral representation 
\be
\label{35b}
\zeta(s)=\frac{2\pi^\frac{s}{2}}{\Gamma\left(\frac{s}{2}\right)}\int\limits_{0}^{\infty}x^{s-1}\omega\left(x^2\right)\ud x,\quad\re s>1,
\ee
we obtain for the Riemann zeta function 
\be
\label{35c}
\zeta(s)=2\pi^{s}\frac{\Gamma(s-1)}{\Gamma\left(\frac{s}{2}\right)}+\frac{2\pi^\frac{s}{2}}{\Gamma\left(\frac{s}{2}\right)}\int\limits_{0}^{\infty}x^{s-1}\left[\omega\left(x^2\right)-\frac{\ue^{-x}}{2x}\right]\ud x,\quad\re s>0.
\ee
From \eref{35a} and \eref{35c} we conclude
\be
\label{35d}
\In_f^{**}=\frac{1}{2}\underset{s\rightarrow 1}{\lim}\left[\zeta(s)-\pi^\frac{s}{2}\frac{\Gamma(s-1)}{\Gamma\left(\frac{s}{2}\right)}\right]=\frac{3}{4}\gamma-\frac{1}{2}\ln\left(2\sqrt{\pi}\right).
\ee
Using the Taylor expansion of $e^x$ we get for the asymptotics of ${\T}(t,k)$ by (\ref{35}) for fixed $k$ the following theorem
\begin{theorem}
\label{36a}
The generalized heat kernel ${\T}(t,k)$ possesses for fixed $k\in\kz$ the asymptotics
\begin{equation}
\label{36}
\eqalign{
{\T}(t,k) &\sim-\sum\limits_{n=0}^{\infty}\frac{\sqrt{\pi}}{4}\frac{(-1)^nk^{2n}}{n!}t^{n-\frac{1}{2}}\ln t \\
        &\phantom{\sim} + \sum\limits_{n=0}^{\infty}\frac{\sqrt{\pi}}{4}\left(3\gamma-2\ln 2\right)\frac{(-1)^nk^{2n}}{n!}t^{n-\frac{1}{2}} \\
        &\phantom{\sim} + \sum\limits_{n=0}^{\infty}\frac{(-1)^nk^{2n}}{4n!}t^n, \quad t\rightarrow 0^+.
}
\end{equation}
\end{theorem}
We define analogously to (\ref{31}) the ``generalized spectral zeta function'' for $-\Delta$
\begin{equation}
\label{37}
\eq{
{\Zl}(s,k) &:=\tr\left(\left(-\Delta+k^2\right)^{-s}\right)\\
 &=\sum\limits_{n=1}^{\infty}\frac{\dz(n)}{\left(n^2+k^2\right)^s}, \quad\re s>\frac{1}{2}, \quad k\neq\pm\ui n, \quad n\in\nz, 
}
\end{equation}
(defined by the principal branch of the logarithm) and (see (\ref{32aab}))
\begin{equation}
\label{38}
T(k):={\Zl}(1,k)=\sum\limits_{n=1}^{\infty}\frac{d(n)}{n^2+k^2}, \quad k\neq\pm\ui n, \quad n\in\nz, 
\end{equation}
and observe
\begin{equation}
\label{39}
{\Zl}(s,0) =\sum\limits_{n=1}^{\infty}\frac{\dz(n)}{n^{2s}}= \zeta^2(2s), \quad s\in\kz\setminus\left\{\frac{1}{2}\right\}.
\end{equation}
With the same arguments as in section \ref{30a} for the expression (\ref{32ab}), we obtain for ${\Zl}(s,k)$
\begin{equation}
\label{40}
\eqalign{
{\Zl}(s,k) &=\frac{1}{\Gamma(s)}\widehat{{\T}}(s,k)\\
         &=\frac{1}{\Gamma(s)}\int\limits_{0}^{\infty}t^{s-1}\ue^{-k^2t}\T(t)\ud t,\quad \re s>\frac{1}{2}, \quad \re k^2>-1.
} 
\end{equation}
%
In order to obtain the analytic continuation of ${\Zl}(s,k)$ as a function of $s$ for fixed $k$, we  apply theorem \ref{32a} to the spectral zeta function (\ref{40}) and use at $s=\frac{1}{2}-n$, $n\in\nz$, the Taylor expansion (setting $x=s+\left(n-\frac{1}{2}\right)$, $n\in-\nz_0$)
\begin{equation}
\label{40a}
\hspace{-2cm}\eqalign{
\frac{1}{\Gamma(s)} &=\frac{1}{\Gamma\left(\frac{1}{2}-n+x\right)}\\
            &=\frac{1}{\Gamma\left(\frac{1}{2}-n\right)}  -\frac{\psi\left(\frac{1}{2}-n\right)}{\Gamma\left(\frac{1}{2}-n\right)}\left(s+\left(n-\frac{1}{2}\right)\right)+\Or\left(\left(s+\left(n-\frac{1}{2}\right)\right)^2\right).
}
\end{equation}
We infer the following theorem
\begin{theorem}
\label{41}
${\Zl}(s,k)$, defined in $\eref{37}$, $k\in\left\{k\in\kz;\quad |\im k|<\sqrt{1+(\re k)^2}\right\}$ fixed, has as a function of $s$ a meromorphic extension to all of $\kz$ with poles at $s_n=-\left(n-\frac{1}{2}\right)$, $n\in\nz_0$, and respective principal parts 
\begin{equation}
\label{42}
\hspace{-2.5cm}\eqalign{
\left(\Pp \left[{\Zl}\right]_{s_n}\right)(s,k) &=\phantom{+}\frac{\sqrt{\pi}(-1)^nk^{2n}}{4\Gamma\left(\frac{1}{2}-n\right)n!}\frac{1}{\left(s+\left(n-\frac{1}{2}\right)\right)^2}\\
&\hspace{-1.5cm}+\left[\left(3\gamma-\frac{1}{2}\ln 2\right)\frac{\sqrt{\pi}(-1)^nk^{2n}}{4\Gamma\left(\frac{1}{2}-n\right)n!}-\frac{\sqrt{\pi}\psi\left(\frac{1}{2}-n\right)}{4\Gamma\left(\frac{1}{2}-n\right)}\frac{(-1)^nk^{2n}}{n!}\right]\frac{1}{s+\left(\frac{1}{2}-n\right)}.
}
\end{equation}
Furthermore, the function ${\Zl}(s,0)$ has ``trivial'' zeros at $s=-n$, $n\in\nz$ (see $\eref{39}$).  
\end{theorem} 
A direct consequence of theorem \ref{41} is that for $k=0$ the poles at $s_n=-\left(n-\frac{1}{2}\right)$, $n\in\nz$, vanish in agreement with (\ref{39}). Furthermore, with the Laurent expansion \cite[p.\,807]{Abramowitz:1992} at $s=1$ (see \eref{103})
\begin{equation}
\label{23}
\zeta(s)=\frac{1}{s-1}+\gamma-\gamma_1(s-1)+\Or\left((s-1)^2\right),
\end{equation}
%
we get for ${\Zl}(s,0)$ at $s=\frac{1}{2}$ the Laurent expansion
\begin{equation}
\label{44}
{\Zl}(s,0)=\zeta^2(2s)=\frac{1}{4\left(s-\frac{1}{2}\right)^2}+\frac{\gamma}{s-\frac{1}{2}}+\gamma^2-2\gamma_1+\Or \left(s-\frac{1}{2}\right).
\end{equation}

We can compute $\Zl(s,k)$ also from the trace formula (\ref{13}) using the function $f_k(s,p):=\frac{1}{\left(p^2+k^2\right)^s}$, $\re s>\frac{1}{2}$, $k\neq\pm\ui n$, $n\in\nz$, which fulfills as a function of $p$ the requirements of theorem \ref{12}. Due to the identity (\ref{13}) we split ${\Zl}(s,k)$,
\begin{equation}
\label{47}
{\Zl}(s,k):={\Zl^W}(s,k)+{\Zl^\mathrm{po}}(s,k),
\end{equation}
and define analogously to (\ref{25}) a ``Weyl term'' 
\begin{equation}
\label{48}
{\Zl^W}(s,k):=\int\limits_0^{\infty}(\ln p+2\gamma)f_k(s,p)\ud p +\frac{f_k(s,0)}{4}, \\
\end{equation}
and a ``periodic orbit term''
\begin{equation}
\label{49}
\fl{\Zl^\mathrm{po}}(s,k):=4\sum\limits_{n=1}^{\infty}\int\limits_0^{\infty}d(n)\left[K_0\left(2(2\pi)^{\frac{3}{2}}\sqrt{\frac{k}{\lf_{p,n}}}\right)-\frac{\pi}{2}Y_0\left(2(2\pi)^{\frac{3}{2}}\sqrt{\frac{k}{\lf_{p,n}}}\right)\right]f_k(s,p)\ud p.
\end{equation}
In order to evaluate ${\Zl^W}(s,k)$ we calculate the integrals of each summand in (\ref{48}) separately. We obtain for the integral in the first summand \cite[p.\,527]{Prudnikov:1986b}  
\begin{equation}
\label{50}
\eqalign{
\int\limits_0^{\infty}\frac{\ln p}{\left(p^2+k^2\right)^s}\ud p &= 
\frac{1}{\Gamma(s)}\int\limits_{0}^{\infty}t^{s-1}\ue^{-k^2t}\left[\int\limits_{0}^{\infty}\ue^{-p^2t}\ln p \ud p\right]\ud t\\
& =\frac{\sqrt{\pi}}{4\Gamma(s)}\int\limits_{0}^{\infty}t^{s-\frac{3}{2}}\ue^{-k^2t}\left[\psi\left(\frac{1}{2}\right)-\ln t\right]\ud t \\
& =\frac{\sqrt{\pi}\Gamma\left(s-\frac{1}{2}\right)}{4\Gamma(s)k^{2s-1}}\left[\psi\left(\frac{1}{2}\right)-\psi\left(s-\frac{1}{2}\right)+2\ln k\right],\\
 & \hspace{1cm} \re s>\frac{1}{2}, \quad \re k^2 >0,
}
\end{equation}
and for the second summand with a similar calculation \cite[p.\,527]{Prudnikov:1986b}
\begin{equation}
\label{51}
\int\limits_0^{\infty}\frac{2\gamma}{\left(p^2+k^2\right)^s}\ud p=\frac{\gamma\sqrt{\pi}\Gamma\left(s-\frac{1}{2}\right)}{\Gamma(s)k^{2s-1}}, \quad \ \re s>\frac{1}{2}, \quad \re k^2 >0.
\end{equation}
Therefore, we get for ${\Zl^{W}}(s,k)$ ($\psi\left(\frac{1}{2}\right)=-\gamma-2\ln 2$) 
\begin{equation}
\label{52}
\eqalign{
{\Zl^W}(s,k) & =\frac{\sqrt{\pi}\Gamma\left(s-\frac{1}{2}\right)}{4\Gamma(s)k^{2s-1}}\left[-\psi\left(s-\frac{1}{2}\right)+\ln\left(\frac{k^2}{4}\right)+3\gamma\right]+\frac{1}{4k^{2s}} \\
 & \hspace{1cm} \re s>\frac{1}{2}, \quad \re k^2 >0.
}
\end{equation}
Of course, with the representation (\ref{52}), the function ${\Zl^W}(s,k)$ can as a function of $s$ ($\re k^2>0$, fixed) be analytically continued to $\kz\setminus\left\{s=-\left(n-\frac{1}{2}\right); \quad n\in\nz_0\right\}$. Using the Taylor expansion (\ref{40a}) at $s_0=\frac{1}{2}$ and \cite[pp.\,3, 13]{Magnus:1966}
\begin{equation}
\label{53}
\eqalign{
\psi\left(s-\frac{1}{2}\right) &=-\frac{1}{s-\frac{1}{2}}-\gamma+\Or\left(s-\frac{1}{2}\right),\\
\Gamma\left(s-\frac{1}{2}\right) &=\frac{1}{s-\frac{1}{2}}-\gamma+\Or\left(s-\frac{1}{2}\right),
}
\end{equation}
we obtain from (\ref{52}) for the principal part of ${\Zl^W}(s,k)$ at $s_0=\frac{1}{2}$ ($\re k^2>0$, fixed)
\begin{equation}
\label{54}
\Pp\left[{\Zl^W}\right]_{\frac{1}{2}}(s,k)=\frac{1}{4\left(s-\frac{1}{2}\right)^2}+\frac{\gamma}{s-\frac{1}{2}}+\Or(1),
\end{equation} 
in complete agreement with (\ref{44}) (see (\ref{39})).

We want to investigate the holomorphy of ${\Zl}(s,k)$ defined in (\ref{37}) as a function on $\kz^2$ in $s$ and $k$. For this reason, we again apply Hartogs' theorem \ref{55a}. Due to the normal convergence of the series in (\ref{37}) the function ${\Zl}(s,k)$,  $\re s >\frac{1}{2}$ fixed, is as a function of $k$ holomorphic on $\kz\setminus\left\{k\in\kz;\quad |\im k|\geq1,\quad \re k=0 \right\}$. On the other hand we know from theorem \ref{41} that ${\Zl}(s,k)$, $k\in\left\{k\in\kz;\quad |\im k|<\sqrt{1+(\re k)^2}\right\}$ fixed, is as a function of $s$ holomorphic on $\kz\setminus\left\{z_n=-\left(n-\frac{1}{2}\right); \quad n\in\nz_0\right\}$. Therefore, we infer \cite[p.\,135]{Cartan:1995} 
\begin{theorem}
\label{60}
${\Zl}(s,k)$ is an analytic function on 
\begin{equation}
\label{60a}
\hspace{-2cm}S_0\times \widetilde{K} :=\left\{(s,k)\in\kz\times\kz; \quad \re s>\frac{1}{2}, \quad k\in\kz\setminus\left\{|\im k|\geq1,\quad \re k=0 \right\}\right\}.
\end{equation}
\end{theorem}
Thus, by the above results for ${\Zl}(s,k)$ we can apply theorem \ref{55a} and get ($K\subset\widetilde{K}$)
\begin{cor}
\label{61}
${\Zl}(s,k)$ is an analytic function on 
\begin{equation}
\label{61a}
\hspace{-2.5cm}S\times K :=\left\{(s,k)\in\kz\times\kz; \quad s\neq-\left(n-\frac{1}{2}\right), \ n\in\nz_0,\quad|\im k|<\sqrt{1+(\re k)^2}\right\}.
\end{equation}
\end{cor}
\begin{rem}
\label{900a}
We want to remark that the domains of $k$ in corollary $\ref{808}$ respectively corollary $\ref{61}$ can be enlarged to the lager domains $\widetilde{\mathcal{K}}$ respectively $\widetilde{K}$ (see $\cite{Lang:1993}$).
\end{rem}
\hspace{-9mm}But we can't directly apply the result of \cite[p.\,17]{Lang:1993} ($\lambda_1=0$, $a_0\neq0$ is required in the series $\sum\limits_{n=0}^{\infty}a_k\ue^{\lambda_kt}$). We omit the proof of remark \ref{900a} (it's an easy adaption of the proof in \cite{Lang:1993}) since we don't need this generalization.
\section{A ``secular equation''}
\label{55}
For our simple infinite quantum graph $\mathfrak{G}$, the eigenvalues respectively wavenumbers are explicitly known. However, in order to investigate the eigenvalues of a more complex infinite quantum graph, it would be interesting to find a ``secular equation`` determining the eigenvalues for the quantum mechanical system $(-\Delta,\Dz(\Delta))$ (see (\ref{4}) and (\ref{5})). We recall that for a compact quantum graph with $E$ edges, $0<E<\infty$, and with the $\delta$-type boundary conditions discussed in the introduction this was fist done in \cite{KottosSmilansky:1998} and generalized for general self adjoint boundary conditions by \cite{KostrykinSchrader:2006b}. Note that the representations used in \cite{KottosSmilansky:1998} and \cite{KostrykinSchrader:2006b} differ only by a cyclic permutation of the matrix entries. In the setting of \cite{KottosSmilansky:1998,KostrykinSchrader:2006b} the secular equation $F_E(k)$ for a compact quantum graph is a function defined by a determinant
\begin{equation}
\label{56}
F_E(k):=\det\left(\eins_{2E\times2E}-S_{2E}(k)T_{2E}(k)\right),  \quad k\in \rz,
\end{equation} 
and the nonzero real zeros $k_n$ of $F_E(k)$ are exactly the wavenumbers (and the orders of the zeros are exactly the multiplicities of the corresponding eigenvalues) of the corresponding quantum graph. The matrix $S(k)$ is called the S-matrix and describes the quantum mechanical scattering at the vertices (see \cite{KostrykinSchrader:1999,KostrykinSchrader:2006} for a detailed discussion). The matrix $T_{2E}$ can be interpreted as a ``transfer'' matrix \cite{KottosSmilansky:1998} describing the propagation of the eigenfunctions along the edges. It has in the setting of \cite{KostrykinSchrader:2006b} the form
\begin{equation}
\label{57}
T_{2E}(k):=
\left(
\begin{array}{cc}
0 & \ue^{\ui k\vl} \\
\ue^{\ui k\vl} & 0
\end{array}
\right)
, \quad 
\left(\ue^{\ui k\vl}\right)_{mn}:=\delta_{mn}\ue^{\ui k l_n},\quad 1\leq m,n\leq E,
\end{equation}
where $0<l_n<\infty$ denotes the length of the edge $e_n$. Clearly, the definition (\ref{57}) doesn't make sense for an infinite quantum graph. Therefore, we make a rearrangement of the matrix entries and define for a compact quantum graph
\begin{equation}
\label{58}
\eqalign{
\widetilde{T}_{2E}(k) & :=
\left(
\begin{array}{ccc}
T^1(k) & & 0 \\
& \ddots & \\
0 & & T^E(k)
\end{array}
\right)
,\quad T^n(k):=
\left(
\begin{array}{cc}
0 & \ue^{\ui kl_n} \\
\ue^{\ui kl_n} & 0
\end{array}
\right)
, \\
 & \hspace{7cm} \quad 1\leq n\leq E.
}
\end{equation}
We remark that the rearrangement can be achieved by a permutation matrix $P$ by $\widetilde{T}_{2E}(k)=P^+T_{2E}(k)P$. The S-matrix of a compact quantum graph with $E$ edges is for Dirichlet boundary conditions (``hard wall reflection'') in the setting of \cite{KostrykinSchrader:2006b} (minus) the unit matrix, $S_{2E}(k)=-\eins_{2E\times2E}$, and thus $\widetilde{S}_{2E}(k):=P^+S_{2E}(k)P=-\eins_{2E\times2E}$. Therefore, the secular equation in (\ref{56}) is for Dirichlet boundary conditions equivalent to (compact graph)
\begin{equation}
\label{59}
\widetilde{F}_E^D(k):=\det\left(\eins_{2E\times2E}+\widetilde{T}_{2E}(k)\right).
\end{equation}
We now extend this rearrangement to our infinite quantum graph $\mathfrak{G}$ (see section \ref{z15}) and observe that the eigenvalue problem $\widetilde{T}_{\infty}(k)\phi=\phi$ with $\phi\in l^2$ ($\widetilde{T}_{\infty}(k)$ as a linear and bounded operator on the Hilbert space of square summable sequences) yields exactly the wavenumbers $k_{n,m}=nm$. However, if we would take over the definition of $\widetilde{F}_E^D(k)$ in (\ref{59}) for our infinite quantum graph with Dirichlet boundary conditions, there would arise a problem since  
\begin{equation}
\label{59a}
\hspace{-1.5cm}\widetilde{F}^D_{\infty}(k):=``\det\left(\eins+\widetilde{T}_{\infty}(k)\right)":=\prod\limits_{m=1}^{\infty}\left(1-\ue^{\frac{2\pi\ui k}{m}}\right)=0, \quad \mbox{for all} \quad k\in\kz.
\end{equation}    
Therefore, we have to seek a proper regularization of $\widetilde{F}^D_{\infty}(k)$ in such a way that the regularized function $\widetilde{F}^{D,R}_{\infty}(k)$ is given by a convergent infinite product and vanishes at precisely the wavenumbers $\pm\ui k_{n,m}$ of our quantum graph $\mathfrak{G}$. We remark that the operator $\widetilde{T}_{\infty}(k)$ is not an element of any trace ideal $T_p$, $p\leq1$, and therefore we cannot apply results of \cite{Simon:1979} for regularization.

There are different methods to obtain a proper regularized secular equation for the infinite quantum graph $\mathfrak{G}$. One elegant method, applicable also in more general situations, is to use the method of zeta function regularization. This method will be applied in the next section, but first we shall use an alternative method which is well adopted to the very special situation of the infinite quantum graph $\mathfrak{G}$ studied in this paper. To this purpose let us define the function ($D(0)=\frac{1}{2\pi}$)  
\begin{equation}
\label{60b}
D\left(k^2\right):=\frac{1}{2\pi}\prod\limits_{n=1}^{\infty}\left(1+\frac{k^2}{n^2}\right)^{\dz(n)},\quad k\in\kz.
\end{equation}
Since the Weierstraß product in \eref{60b} is normally convergent on $\kz$, we obtain the
\begin{cor}
\label{60c}
The function $D\left(k^2\right)$ defined in $\eref{60b}$ is an entire function on $\kz$ whose zeros $z_l$, $l\in\nz$, are exactly located at the wavenumbers $\pm \ui k_{mn}=\pm\ui mn$ ($m,n\in \nz$) of the quantum graph $\mathfrak{G}$. The order of the zero at $z_l=\pm\ui l$ is given by the multiplicity $d(k_l)=d(l)$ of the associated wavenumber $k_l=l$.
\end{cor}
Rewriting the infinite product in \eref{60b} as a double product,
\be
\label{60d}
D\left(k^2\right)=\frac{1}{2\pi}\prod\limits_{n=1}^{\infty}\prod\limits_{m=1}^{\infty}\left(1+\frac{(k/n)^2}{m^2}\right),\quad k\in\kz,
\ee
and using the identity
\be
\label{60e}
\prod\limits_{m=1}^{\infty}\left(1+\frac{z^2}{m^2}\right)=\frac{\sinh(\pi z)}{\pi z},\quad z \in\kz,
\ee
we obtain
\begin{theorem}
\label{60f}
The entire function $D\left(k^2\right)$ defined in $\eref{60b}$ possesses the following factorization in terms of the periodic orbits $\lf_{p,n}=\frac{2\pi}{n}$ of the infinite quantum graph  $\mathfrak{G}$   
\be
\label{60g}
\eq{
D\left(k^2\right) &=\frac{1}{2\pi}\prod\limits_{n=1}^{\infty}\left[\frac{\sinh\left(\frac{\pi k}{n}\right)}{\left(\frac{\pi k}{n}\right)}\right]\\
&=\frac{1}{2\pi}\prod\limits_{n=1}^{\infty}\left[\frac{\exp\left(\frac{k\lf_{p,n}}{2}\right)}{k\lf_{p,n}}\left(1-\exp\left(-k\lf_{p,n}\right)\right)\right], \quad k\in\kz.
}
\ee
\end{theorem}
We remark that the products in \eref{60g} are well defined since
\begin{itemize}
	\item $\sinh x=x+\Or\left(x^3\right)$ respectively $1-\ue^{-x}=x+\Or\left(x^2\right)$ for $x\rightarrow0$ ($0$ is a removable singularity), and
	\item the infinite products are compact convergent in agreement with the fact that $D\left(k^2\right)$ is an entire function.
\end{itemize}

It is worthwhile to point out that \eref{60g} has the form of an ``Euler product'' (product over the length spectrum of the primitive periodic orbits of the quantum graph $\mathfrak{G}$). We also note
\begin{cor}
\label{60aa}
The entire function $D\left(k^2\right)$ defined in $\eref{60b}$ possesses the following factorization 
\be
\label{60ab}
D\left(k^2\right)=\frac{1}{2\pi}\frac{1}{\widetilde{\Gamma}(\ui k)\widetilde{\Gamma}(-\ui k)}, \quad k\in\kz,
\ee
in terms of the generalized gamma function $\widetilde{\Gamma}(z)$, defined in $\eref{608}$.
\end{cor}
\hspace{-0.85cm}This follows from \eref{60g}, the identity \cite[p.\, 256]{Abramowitz:1992}
\be
\label{60ac}
\frac{\sinh(\pi z)}{\pi z}=\frac{1}{\Gamma(1+\ui z)\Gamma(1-\ui z)},\quad z\in\kz, 
\ee
and \eref{613}.

In view of the identity \eref{60g}, we can propose a regularization of our original function $\widetilde{F}^D_{\infty}$ in \eref{59a}. For this purpose we define a truncation operation for our infinite matrices $\eins$ and $\widetilde{T}_{\infty}(k)$ in (\ref{59a}) (see (\ref{58})) and for the function (\ref{59a}) (see(\ref{59})) by
\begin{equation}
\label{70}
\eqalign{
\left.\eins\right|_N &:=\eins_{2N},\quad  \left.\widetilde{T}_{\infty}\right|_N(k):=\widetilde{T}_{2N}(k), \\
\left.\widetilde{F}^D_{\infty}\right|_N(k) &:={\widetilde{F}_N}^D(k):=\det\left(\left.\eins\right|_N+\left.\widetilde{T}_{\infty}(k)\right|_N\right) \\
& =\prod\limits_{n=1}^{N}\left[1-\ue^{\frac{2\pi\ui k}{n}}\right], \quad N\in\nz.
}
\end{equation}
We remark that these truncation operations correspond in our case to a truncation of the quantum graph which means that the truncated matrices respectively the truncated functions describe a finite quantum graph with $N$ edges which have the lengths $l_n=\frac{\pi}{n}$ and are equipped with Dirichlet boundary conditions. Since as previously mentioned the zeros $z_l$, $l=1,\ldots,\infty$, of $D\left(k^2\right)$ are exactly located at the wavenumbers $\pm \ui k_{n,m}$ and the orders of the zeros and the multiplicities of the wavenumbers coincide, we define, due to theorem \ref{60f}, the regularized function $\widetilde{F}^{D,R}_{\infty}(k)$ (see (\ref{60g}), $\lf_{p,n}=\frac{2\pi}{n}$)
\begin{equation}
\label{70a}
\widetilde{F}^{D,R}_{\infty}(k):=\lim_{N\rightarrow\infty}\frac{1}{2\pi}\left(\frac{N!\ue^{k\mathcal{L}_{N}}}{\left(2\pi k\right)^{N}}\left.\widetilde{F}^D_{\infty}\right|_N(\ui k)\right)=D\left(k^2\right), \quad k\in \kz, 
\end{equation}
where $\mathcal{L}_{N}:=\sum\limits_{n=1}^{N}l_n$ is the total length of the truncated quantum graph consisting of $N$ edges. Due to (\ref{70}) this can be regarded as a regularized limit of the secular equation (\ref{59}) for finite quantum graphs with Dirichlet boundary conditions.

\section{The determinant of $-\Delta+k^2$}
\label{t1}

As already mentioned, there is another general regularization method which uses the generalized zeta function $\Zl(s,k)$ which we have already introduced in eq. \eref{37}. This method was recently applied in \cite{Harrison:2009} to obtain the determinant of a finite quantum graph. Our starting point is the functional determinant of the elliptic operator $-\Delta+k^2$ which is defined as follows
\begin{equation}
\label{62}
\eq{
\Df\left(k^2\right)&:=\det\left(-\Delta+k^2\right):=\exp\left(\left.-\frac{\partial{\Zl}(s,k)}{\partial s}\right|_{s=0}\right), \quad k\in K,\\
& \hspace{3.8cm} K=\left\{k\in\kz; \ |\im k|<\sqrt{1+(\re k)^2}\right\}.
}
\end{equation}
This definition is in accordance with \cite{Singer:1973} and has first been applied to quantum field theory in \cite{Hawking:1977}. Note that $\left.-\frac{\partial{\Zl}(s,k)}{\partial s}\right|_{s=0}$ is well defined due to corollary \ref{61}. We will prove the following 
\begin{theorem}
\label{63}
The functional determinant $\eref{62}$ has an analytic continuation on $\kz$ as an entire function and agrees with the Weierstraß product $D\left(k^2\right)$ defined in $\eref{60b}$, i.e. it holds
\be
\label{64}
\Df\left(k^2\right)=\det\left(-\Delta+k^2\right)=\frac{1}{2\pi}\prod\limits_{n=1}^{\infty}\left(1+\frac{k^2}{n^2}\right)^{d(n)}
\ee
with $D(0)=\Df(0)=\det(-\Delta)=\frac{1}{2\pi}$.
\end{theorem}
\begin{proof}
In order to show $\Df\left(k^2\right)=D\left(k^2\right)$ our strategy is to prove this in a neighbourhood $U_\delta(0)$, $\delta>0$, at $k=0$ and then by analytical continuation to conclude (\ref{64}) on $\kz$. Therefore, we calculate (see (\ref{37}) [principal branch of the logarithm]) 
\begin{equation}
\label{65}
\frac{\partial{\Zl}(s,k)}{\partial s}=-\sum\limits_{n=1}^{\infty}d(n)\frac{\ln\left(k^2+n^2\right)}{\left(k^2+n^2\right)^{s}} 
, \quad \re s>\frac{1}{2},\quad k\in K.
\end{equation}
Furthermore, we get
\begin{equation}
\label{66}
\eqalign{
\frac{\partial^2{\Zl}(s,k)}{\partial k\partial s} = & -2k\sum\limits_{n=1}^{\infty}\frac{d(n)}{\left(k^2+n^2\right)^{s+1}} \\
                                                     & +2sk\sum\limits_{n=1}^{\infty}d(n)\frac{\ln\left(k^2+n^2\right)}{\left(k^2+n^2\right)^{s+1}}
, \quad \re s>-\frac{1}{2}, \quad k\in K,
}
\end{equation}
where in the last line we have used the uniqueness of the analytic continuation. Setting $s=0$ (see (\ref{66}) and (\ref{32aab})), we obtain 
\begin{equation}
\label{67}
-\frac{\partial\phantom{k}}{\partial k}\left[\left.\frac{\partial{\Zl}(s,k)}{\partial s}\right|_{s=0}\right]=2k\sum\limits_{n=1}^{\infty}\frac{d(n)}{k^2+n^2}=2kT(k),\quad |k|<1.
\end{equation}
On the other hand we have (see (\ref{60b}) and (\ref{32aab})) 
\begin{equation}
\label{68}
\hspace{-2cm}\dk\ln D\left(k^2\right) = 2k\sum\limits_{n=1}^{\infty}\frac{d(n)}{k^2+n^2}=2kT(k)=-\frac{\partial\phantom{k}}{\partial k}\left[\left.\frac{\partial{\Zl}(s,k)}{\partial s}\right|_{s=0}\right],\quad |k|<1.
\end{equation}
By (\ref{39}) and \cite[p.\,807]{Abramowitz:1992} ($\zeta'(0)=-\frac{1}{2}\ln(2\pi)$, $\zeta(0)=-\frac{1}{2}$) we get 
\begin{equation}
\label{69}
\eqalign{
\exp\left(-\left.\frac{\partial{\Zl}(s,0)}{\partial s}\right|_{s=0}\right) & =\exp\left(-4\zeta'(0)\zeta(0)\right)=\frac{1}{2\pi} \\
 & =D(0)=\Df(0)=\det\left(-\Delta\right).
 }
\end{equation}
\end{proof}
%
%
%
%
%
%
%
%
%
\section{The trace of the resolvent of $-\Delta$}
\label{32aaa}
First of all, we want to express the trace of the resolvent 
of $-\Delta$ in terms of the primitive periodic orbits of the quantum graph $\mathfrak{G}$ (see section \ref{z15}). Explicitly, the trace of the resolvent is given by (see also \eref{38})
\begin{equation}
\label{32aab}
\hspace{-2.2cm} T(k):=\tr\left(-\Delta+k^2\right)^{-1}=\sum\limits_{l=1}^{\infty}\frac{\dz(l)}{l^2+k^2}=\sum\limits_{n=1}^{\infty}\sum\limits_{m=1}^{\infty}\frac{1}{(nm)^2+k^2}, \quad k\neq\pm\ui l, \ l\in\nz.
\end{equation}
Using the identity \cite[p.\,471]{Magnus:1966} 
\begin{equation}
\label{32aac}
\sum\limits_{m=0}^{\infty}\frac{1}{m^2x^2+y^2}=\frac{1}{2y^2}+\frac{\pi}{2xy}\coth\left(\frac{\pi y}{x}\right)
\end{equation}
and replacing $n$ by $\frac{2\pi}{\lf_{p,n}}$ ($\lf_{p,n}=\frac{2\pi}{n}$ is the length of the $n^{th}$ primitive classical periodic orbit of the graph $\mathfrak{G}$), we obtain from (\ref{32aab}) the following ``periodic orbit series'' of the resolvent
\begin{equation}
\label{32aad}
T(k)=\sum\limits_{n=1}^{\infty}\left[-\frac{1}{2k^2}+\frac{\lf_{p,n}}{4k}\coth\left(\frac{k\lf_{p,n}}{2}\right)\right], \quad k\neq\pm\ui l, \ l\in\nz.
\end{equation}
Using the Laurent expansion of $\coth x$ at $x=0$
\begin{equation}
\label{32aadaaa}
\coth x=\frac{1}{x}+\sum\limits_{n=1}^{\infty}\frac{2^{2n}B_{2n}}{(2n)!}x^{2n-1}, \quad 0<|x|<\pi,
\end{equation}
we get from (\ref{32aad}) for $T(k)$ (the summations can be interchanged due to the absolute convergence of the series) the power series
\begin{equation}
\label{32aae}
\eqalign{
T(k) & =\frac{1}{2k^2}\sum\limits_{n=1}^{\infty}\sum\limits_{m=1}^{\infty}\frac{2^{2m}B_{2m}}{(2m)!}\frac{(\pi k)^{2m}}{n^{2m}} \\
     & =\frac{1}{2k^2}\sum\limits_{m=1}^{\infty}\left[\frac{2^{2m}B_{2m}}{(2
     m)!}(\pi k)^{2m}\left(\sum\limits_{n=1}^{\infty}\frac{1}{n^{2m}}\right)\right] \\
     & =\frac{1}{2k^2}\sum\limits_{m=1}^{\infty}\frac{2^{2m}B_{2m}\zeta(2m)}{(2m)!}(\pi k)^{2m}\\
     & =\sum\limits_{m=0}^{\infty}(-1)^m\zeta^2(2m+2)k^{2m}, \quad 0\leq|k|<1,
}
\end{equation}
where in the last line we used the relation \cite[p.\,19]{Magnus:1966}
\be
\label{2000}
B_{2m}=\frac{2(2m)!(-1)^{m+1}\zeta(2m)}{(2\pi)^{2m}},\quad m\in\nz_0.
\ee
Note that \cite[p.\,19]{Magnus:1966}
\be
\label{300}
T(0)=\tr\left(-\Delta\right)^{-1}=\Zh(2,0)=\sum\limits_{n=1}^{\infty}\frac{\dz(n)}{n^2}=\zeta^2(2)=\left(\frac{\pi^2}{6}\right)^2.
\ee
We get the following theorem:
\begin{theorem}
\label{z300}
The trace of the resolvent $T(k)$ possesses for $0\leq|k|<1$ the absolutely convergent series expansion $\eref{32aae}$.
\end{theorem}
We want to derive an alternative expression for $\Tho(t)$, defined in \eref{27}, involving the trace of the resolvent $T(k)$. For this we use the relation \cite[p.\,345]{Magnus:1966}
\be
\label{301}
\int\limits_{0}^{\infty}\frac{x\cos(ax)}{x^2+y^2}\ud x=\frac{1}{2}\left[\ue^{ay}\Ee(ay)-\ue^{-ay}\Es(ay)\right],\quad a,y>0,
\ee
and insert it in the expression \eref{27}. 
%
%
We then obtain
\be
\label{302}
\Tho(t)=\frac{4}{t}\sum\limits_{n=1}^{\infty}d(n)\int\limits_{0}^{\infty}\frac{k\cos\left(\frac{4\pi^2k}{t}\right)}{k^2+n^2}\ud k\quad t\in\rz_{>0}.
\ee
Note that due to the asymptotics \eref{b6} we cannot interchange the summation with the improper Riemann integral.

For the asymptotics of $T(k)$ for $k\rightarrow\infty$ we evaluate \eref{47} at $s=1$ and obtain the trace of the resolvent \eref{32aab} 
\be
\label{r1a}
\eq{
T(k)={\Zl}(1,k)=\frac{1}{2k}\frac{\Df'\left(k^2\right)}{\Df\left(k^2\right)}&={\Zl^W}(1,k)+{\Zl^\mathrm{po}}(1,k)\\
&=:T^W(k)+T^{\po}(k),\quad k>0,
}
\ee
with $\Df\left(k^2\right)$ defined in \eref{62} (see theorem \ref{63}). From \eref{52} we obtain for the ``Weyl term''
\be
\label{r1}
\eq{
T^W(k) &=\frac{1}{2k}\left[\pi\ln k +2\pi\gamma+\frac{1}{2k}\right]\\
       &=\frac{1}{2k}\dk\ln\left[\exp\left(\pi k\ln k+(2\gamma-1)\pi k+\frac{1}{2}\ln k\right)\right],\quad k>0,
}
\ee
where the logarithm is defined by its principal branch. According to \eref{49} $T^{\po}$ is given by (with the primitive periodic orbit lengths $\lfp=\frac{2\pi}{n}$ [see section \ref{9}])
\be
\label{r2}
\hspace{-2cm}\eq{
T^{\po}(k) &=2\pi\sum\limits_{n=1}^{\infty}d(n)\int\limits_{0}^{\infty}\left[\frac{2}{\pi}K_0\left(4\pi\sqrt{nk'}\right)-Y_0\left(4\pi\sqrt{nk'}\right)\right]\frac{\ud k'}{k'^2+k^2}\\
&=4\pi\sum\limits_{n=1}^{\infty}d(n)\int\limits_{0}^{\infty}\left[\frac{2}{\pi}K_0\left(4\pi\sqrt{n}k'\right)-Y_0\left(4\pi\sqrt{n}k'\right)\right]\frac{k'\ud k'}{k'^4+k^2}\\
&=\frac{4\pi}{k}\sum\limits_{n=1}^{\infty}d(n)\ker\left(4\pi\sqrt{nk}\right)=\frac{4\pi}{k}\sum\limits_{n=1}^{\infty} d(n)\re K_0\left(4\pi\sqrt{\ui nk}\right)\\
&=\frac{4\pi}{k}\sum\limits_{m=1}^{\infty}\sum\limits_{n=1}^{\infty}\re K_0\left(2(2\pi)^{\frac{3}{2}}\sqrt{\ui \frac{mk}{\lfp}}\right),\quad k>0,
}
\ee
where the integral has been expressed in terms of the Kelvin function $\ker(z)$  (see \cite[p.\,373]{Prudnikov:1986}). Formula \eref{r2} expressed in terms of the $K_0$ Bessel function has also been given in \cite{Oberhettinger:1972}. The last step in \eref{r2} is justified due to the asymptotics \eref{rr3}. Thus, the double series in \eref{r2} is absolutely and locally uniformly convergent. Due to \eref{rr3} we can estimate the penultimate line in \eref{r2} as  ($k\rightarrow\infty$)
\be
\label{rar2}
\hspace{-1.5cm}\eq{
T^{\po}(k)&=\frac{4\pi}{k}\sum\limits_{n=1}^{\infty} d(n)\re K_0\left(4\pi\sqrt{\ui nk}\right)\\
&=\sqrt{2}\pi k^{-\frac{5}{4}}\sum\limits_{n=1}^{\infty} \frac{d(n)}{n^{\frac{1}{4}}}\re\left[\ui^{-\frac{1}{4}}\exp\left(-4\pi\sqrt{\ui nk}\right)\left(1+\Or\left(\frac{1}{\sqrt{nk}}\right)\right)\right],
} 
\ee
and we notice that the leading asymptotic part of \eref{rar2} for $k\rightarrow\infty$ is given by the very first summand in the series on the r.h.s in \eref{rar2}. For this first summand we obtain 
%
\be
\label{rr4}
\hspace{-1.5cm}\eq{
& \sqrt{2}\pi k^{-\frac{5}{4}}\re\left[\ui^{-\frac{1}{4}}\exp\left(-4\pi\sqrt{\ui k}\right)\right]\\
&\hspace{3cm}=\sqrt{2}\pi k^{-\frac{5}{4}}\exp\left(-2\pi\sqrt{2k}\right)\cos\left(2\pi\sqrt{2k}+\frac{\pi}{8}\right),\quad k>0,
}
\ee
and the remaining series in \eref{rar2} without the summand \eref{rr4}, which converges absolutely, tends faster to zero for $k\rightarrow\infty$ as the leading summand \eref{rr4}. Combining the above result \eref{r1a} with \eref{r1}, we deduce the following theorem: 
\begin{theorem}
\label{ra1}
The trace of the resolvent $T(k)$, defined in $\eref{32aab}$, possesses for $k\rightarrow\infty$ the asymptotics
\be
\label{b6}
\fl\eq{
T(k) &=\frac{\pi\ln k}{2k}+\frac{\pi\gamma}{k}+\frac{1}{4k^2}\\
&\hspace{2cm}+\sqrt{2}\pi k^{-\frac{5}{4}}\exp\left(-2\pi\sqrt{2k}\right)\cos\left(2\pi\sqrt{2k}+\frac{\pi}{8}\right)\left(1+\Or\left(\frac{1}{\sqrt{k}}\right)\right).
}%
\ee%
\end{theorem}

%
%
%
%

%
%
%
%
%
%
%
%
%
%
%
\section{A Selberg-like zeta function satisfying the analogue of the Riemann hypothesis}
\label{ar1}
In this section we construct a spectral function in terms of the ``periodic orbit term'' \eref{r2}, which resembles in its main properties those of the famous Selberg zeta function \cite{Steiner:1987} and of the semiclassical dynamical zeta function well-known in the theory of quantum chaos, see e.g. \cite{Steiner:1991a}. Our starting point is equation \eref{r1a}. Owing to \eref{r1a}, \eref{r1} and \eref{b6} we can define the ``periodic orbit zeta function'' 
\be
\label{r8}
Z(s):=\exp\left(-2\int\limits_{s}^{\infty}kT^{\po}(k)\ud k\right), \quad s>0.
\ee
Because of the existence of the improper Riemann integral in \eref{r8}, the relation 
\be
\label{r6}
\lim_{s\rightarrow\infty}Z(s)=1
\ee
is trivial. By \eref{r1a} and \eref{r1} we then obtain the relation 
\be
\label{r9}
\hspace{-1cm}Z(s)=K\exp\left(-\pi s\ln s-(2\gamma-1)\pi s-\frac{1}{2}\ln s\right)\Df\left(s^2\right),\quad s>0,
\ee
where $K$ is a positive constant. 
In order to determine the constant $K$ in \eref{r9}, we need the following proposition (the asymptotics of $\Df\left(s^2\right)$ for $s\rightarrow\infty$ is refined in section \ref{500aaa}):
\begin{prop}
\label{zz1}
The logarithm of the functional determinant $\eref{62}$ (see theorem $\ref{63}$) possesses the asymptotics
\be
\label{509a}
\eq{
\ln \Df\left(k^2\right) &=\ln\det\left(-\Delta+k^2\right)\\
&=\pi k\ln k+\pi(2\gamma-1)k+\frac{1}{2}\ln k +\orr(1)\quad \mbox{for}\quad  k\rightarrow\infty.
}
\ee
\end{prop}
\begin{proof}
For the proof we use the method of Sarnak \cite{Sarnak:1987}. 
We have (see \eref{62})
\be
\label{500a}
\ln\Df\left(k^2\right)=\ln\det\left(-\Delta+k^2\right)=\left.-\frac{\partial{\Zl}(s,k)}{\partial s}\right|_{s=0},
\ee
and thus we have to investigate the derivative of ${\Zl}(s,k)$ at $s=0$ (see \eref{62} [due to corollary \ref{61}, ${\Zl}(s,k)$ is holomorphic at $s=0$ for $k\in K$]). We use the representation \eref{40} for ${\Zl}(s,k)$ and split it in three terms giving a proper ``regularization'' of the integrals at $t=0$ (see theorem \ref{36a} and theorem \ref{60}),
\be
\label{500}
\fl\eq{
{\Zl}(s,k) &=\frac{1}{\Gamma(s)}\int\limits_{0}^{\infty}t^{s-1}\ue^{-k^2t}\T(t)\ud t\\
           &=\frac{1}{\Gamma(s)}\int\limits_{0}^{1}\left[\T(t)-a\frac{\ln t}{t^{\alpha}}-\frac{b}{t^{\beta}}-c\right]t^{s-1}\ue^{-k^2t}\ud t\\
           &+\frac{1}{\Gamma(s)}\int\limits_{1}^{\infty}t^{s-1}\ue^{-k^2t}\T(t)\ud t\\
           &+\frac{1}{\Gamma(s)}\int\limits_{0}^{1}\left[a\frac{\ln t}{t^{\alpha}}+\frac{b}{t^{\beta}}+c\right]t^{s-1}\ue^{-k^2t}\ud t,\quad \re s>\frac{1}{2}, \quad \re k^2>-1,
}
\ee
with (see \eref{35} and \eref{35d})
\be
\label{501}
\hspace{-1cm}a=-\frac{\sqrt{\pi}}{4},\quad \alpha=\frac{1}{2},\quad b=\frac{\sqrt{\pi}}{2}\left(\frac{3}{2}\gamma-\ln 2\right),\quad \beta=\frac{1}{2},\quad\mbox{and}\quad c=\frac{1}{4}.
\ee
Due to the regularization, the first and the second integral on the r.h.s in \eref{500} are holomorphic functions at $s=0$ for $\re k^2>0$ (see \eref{36} and \cite[p.\,92]{Whittaker:1996}). Since 
\be
\label{502}
\left({\Gamma(s)}^{-1}\right)'_{s=0}=-\left(\frac{\psi(s)}{\Gamma(s)}\right)_{s=0}=1\quad \mbox{and}\mbox\quad\left({\Gamma(s)}^{-1}\right)_{s=0}=0,
\ee
we get for the derivative with respect to $s$ of the first and the second integral on the r.h.s. in \eref{500} evaluated at $s=0$ the expression
\be
\label{503}
\int\limits_{0}^{1}\frac{1}{t}\left[\T(t)-a\frac{\ln t}{t^{\alpha}}-\frac{b}{t^{\beta}}-c\right]\ue^{-k^2t}\ud t+ \int\limits_{1}^{\infty}\frac{1}{t}\ue^{-k^2t}\T(t)\ud t.
\ee
Since the integrals in \eref{503} converge uniformly with respect to $k$ and the integrals converge locally uniformly to zero on $(0,1]$ respectively on $[1,\infty)$ for $k\rightarrow\infty$, we can neglect the first and the second integral in \eref{500} for the asymptotics of $\ln \Df\left(k^2\right)$ for $k\rightarrow\infty$. For the third integral in \eref{500} we get by the substitution $t=\frac{t'}{k^2}$ for $\re s>\frac{1}{2}$
\be
\label{504}
\eq{
& \frac{1}{\Gamma(s)}\int\limits_{0}^{1}\left[a\frac{\ln t}{t^{\alpha}}+\frac{b}{t^{\beta}}+c\right]t^{s-1}\ue^{-k^2t}\ud t \\
= & \frac{ak^{2(\alpha-s)}}{\Gamma(s)}\int\limits_{0}^{\infty}t'^{s-\alpha-1}\ln t' \ue^{-t'}\ud t'-\frac{2a k^{2(\alpha-s)}\ln k}{\Gamma(s)}\int\limits_{0}^{\infty}t'^{s-\alpha-1}\ue^{-t'}\ud t'\\
 & +\frac{bk^{2(\beta-s)}}{\Gamma(s)}\int\limits_{0}^{\infty}t'^{s-\beta-1}\ue^{-t'}\ud t'+\frac{ck^{-2s}}{\Gamma(s)}\int\limits_{0}^{\infty} t'^{s-1}\ue^{-t'}\ud t'\\
 & -\frac{1}{\Gamma(s)}\int\limits_{1}^{\infty}\left[a\frac{\ln t}{t^{\alpha}}+\frac{b}{t^{\beta}}+c\right]t^{s-1}\ue^{-k^2t}\ud t.
}
\ee
By a similar argumentation as previously applied to the corresponding integrals in \eref{500}, we can neglect the last integral in \eref{504} for the desired asymptotic evaluation of $\ln \Df\left(k^2\right)$. By the well-known integral representation of the gamma function  
\be
\label{505}
\left(\frac{\ud\phantom{s}}{\ud s}\right)^n\Gamma(s)=\int\limits_{0}^{\infty}t^{s-1}\ue^{-t}(\ln t)^n\ud t,\quad \re s>0,
\ee
we obtain for the first four integrals on the r.h.s. in \eref{504} the expression 
\be
\label{506}
\eq{
 & ak^{2(\alpha-s)}\frac{\Gamma'(s-\alpha)}{\Gamma(s)}-2a k^{2(\alpha-s)}\ln k\frac{\Gamma(s-\alpha)}{\Gamma(s)}+bk^{2(\beta-s)}\frac{\Gamma(s-\beta)}{\Gamma(s)}\\
 & +ck^{-2s},\quad s-\alpha,s-\beta\notin-\nz_0.
}
\ee
After differentiation of \eref{506} with respect to $s$ and a subsequent evaluation at $s=0$, we get by using \cite[p.\,2,14]{Magnus:1966}
\be
\label{507}
\Gamma\left(-\frac{1}{2}\right) =-2\sqrt{\pi},\quad\Gamma'\left(-\frac{1}{2}\right)=2\sqrt{\pi}(\gamma+2\ln 2-2)
\ee
and \eref{501} and \eref{502} the proposition \ref{zz1}.
\end{proof}
Owing to the asymptotics \eref{509a} we infer from \eref{r6} and \eref{r9} (identifying $s=k$) the identity $K=1$ and obtain by corollary \ref{60c} and theorem \ref{63} the following theorem:
\begin{theorem}
\label{r9a}
$Z(s)$, defined in $\eref{r8}$, is an analytic function on $\kz\setminus\rz_{\leq0}$ and is given by (\,$\ln s$ and $\sqrt{s}$ defined by the principal branch of the logarithm)
\be
\label{r10}
Z(s)=\frac{1}{\sqrt{s}}\exp\left(-\pi s\ln s-(2\gamma-1)\pi s\right)\Df\left(s^2\right),\quad s\in\kz\setminus\rz_{\leq0}.
\ee
Furthermore, the zeros of $Z(s)$ are exactly located on the ``critical line'', i.e. at the imaginary wavenumbers $\pm\ui k_n=\pm\ui n$, $n\in\nz$, and the orders of these zeros are given by the multiplicities $d\hspace{-0.07cm}\left(k_n\right)=d(n)$ of the associated wavenumbers $k_n=n$, $n\in\nz$. Thus, the zeta function $Z(s)$ satisfies the analogue of the Riemann hypothesis.  Additionally, the principal branch of the logarithm of $Z(s)$ possesses on the half plane $\re s>0$ the representation $\eref{r7}$, derived below.
\end{theorem}
For an arbitrary $s$ with $\im s>0$ there holds for the principal branch of the logarithm $\ln(-s)=\ln s-\ui \pi$. Thus, we get for $Z(-s)$ with $\im s>0$ 
\be
\label{nr2}
\hspace{-1cm}Z(-s)=\exp\left(\pi\left(s\ln s +(2\gamma-1)s-\ui\pi s+\frac{1}{2}\ui\right)-\frac{1}{2}\ln s\right)\Df\left(s^2\right),\\
\ee
and thus theorem \ref{p1o}
\begin{theorem}
\label{p1o}
\phantom{u}
\begin{enumerate}
 \item[i)] The zeta function $Z(s)$ satisfies the functional equation 
 \be
 \label{p20}
 Z(-s)=\exp\left(\phi(s)\right)Z(s),\quad \im s>0,
 \ee 
 with
 \be
 \label{p30}
 \phi(s):=2\pi\left( s\ln s+(2\gamma-1)s-\ui\frac{\pi}{2}s+\ui\frac{1}{4}\right).
 \ee
 \item[ii)] The limit $\im s\rightarrow0$, $s\rightarrow k\in\rz_{>0}$ of $\eref{p20}$ exists and thus we infer that $\eref{p20}$ holds for $s\in\rz_{>0}$, too. Evaluating $\eref{p20}$ and $\eref{p30}$ on the critical line, we obtain with $\ln(\ui k)=\ln k+\ui\frac{\pi}{2}$, $k>0$, the following version of the functional equation
\be
\label{nr4}
Z(-\ui k)=\ue^{2\pi\ui N^W(k)}Z(\ui k),\quad k>0,
\ee
with $N^W(k)$ the Weyl term of the counting function $N(k)$, defined in $\eref{q12}$.
\end{enumerate}
\end{theorem}
%
%
%
%
%
%
%
%
%
%
%
%
%
From $\eref{nr4}$ and the Schwarz reflection principle we obtain the relation 
\be
\label{r17}
\hspace{-2cm}\ue^{-\ui\pi N^W(k)}Z(-\ui k)=\overline{\ue^{\ui\pi N^W(k)}Z(\ui k)}=\ue^{\ui\pi N^W(k)}Z(\ui k), \quad k>0.
\ee
Thus, we deduce the following corollary:
\begin{cor}
\label{r17a}
On the critical line $\ui \rz$ the Hardy-like function 
\be
\label{r18}
Z_\Delta(k):=\ue^{\ui\pi N^W(k)}Z(\ui k)
\ee
with $N^W(k)$ defined in $\eref{q12}$ takes real values for all $k>0$.
\end{cor}
\hspace{-0.8cm}Note, that these properties are reminiscent of the Hardy $Z$-function in the theory of the Riemann zeta function defined by ($t\in\rz$)
\be
\label{e4}
Z_H(t):=\ue^{\ui\theta(t)}\zeta\left(\frac{1}{2}+\ui t\right),\quad \theta(t):=\arg\left(\Gamma\left(\frac{2\ui t+1}{4}\right)\right)-\frac{\log\pi}{2}t,
\ee
where $\zeta(s)$ is the Riemann zeta function and $\frac{\theta(t)}{\pi}$ describes the ``Weyl term'' of the counting function for the nontrivial Riemann zeros. However, in this case the critical line is given by $\re s=\frac{1}{2}$.

From \eref{64} we obtain, $k\in\rz_{>0}\setminus\nz$,
\be
\label{r19a}
\hspace{-1cm}\eq{\frac{1}{\pi}\lim_{\epsilon\rightarrow0}\im\ln \Df\left(-(k-\ui\epsilon)^2\right)&=\frac{1}{\pi}\lim_{\epsilon\rightarrow0}\im\sum\limits_{n\leq k}d(n)\ln\left(1-\frac{(k-\ui\epsilon)^2}{n^2}\right)\\
&\phantom{=}+\frac{1}{\pi}\lim_{\epsilon\rightarrow0}\im\sum\limits_{n>k}d(n)\ln\left(1-\frac{(k-\ui\epsilon)^2}{n^2}\right)\\
&=\sum\limits_{n\leq k}d(n)=N(k).
}
\ee
For $k=n\in\nz$ we have to investigate
\be
\label{r19aa}
\eq{
f(n-\ui\epsilon)&:=\left(1-\frac{(n-\ui\epsilon)^2}{n^2}\right)=\left(1-\frac{n-\ui\epsilon}{n}\right)\left(1+\frac{n-\ui\epsilon}{n}\right)\\
&=\frac{\ui\epsilon}{n}\left(2-\frac{\ui\epsilon}{n}\right),\quad \epsilon>0.
}%
\ee
Thus, we obtain 
\be
\label{r20}
\lim_{\epsilon\rightarrow 0^+}\arg f(n-\ui\epsilon)=\frac{\pi}{2}+\lim_{\epsilon\rightarrow 0^+}\arg(2-\frac{\ui\epsilon}{n})=\frac{\pi}{2}.
\ee
From \eref{r19a} and \eref{r20} we get
\be
\label{r21}
\lim_{\epsilon\rightarrow 0}\frac{1}{\pi}\im\ln\Df\left(-(k-\ui\epsilon)^2\right)=N(k)-\frac{1}{2}d(k),\quad k\in\rz_{>0},
\ee
with $d(k)$ defined in \eref{q11d}. Note that $\ln Z(s)$ with the logarithm defined by the requirement that $\ln Z(s)$ is real for $s>0$, which coincides with the principal branch of the logarithm on $s>0$, can be analytically continued to the set 
\be
\label{rar7}
\mathfrak{S}:=\kz\setminus\left\{z=n\ui+x;\quad n\in\gz, \quad x\in\rz_{\leq0}\right\}.
\ee
Combining \eref{r21} with theorem \ref{r9a} and by the above results we get the following theorem:
\begin{theorem}
\label{r21a}
For the spectral counting function $N(k)$, defined in $\eref{10}$, the following relations hold
\be
\label{r22}
\eq{
N(k)&=\lim_{\epsilon\rightarrow 0}\frac{1}{\pi}\im\ln\Df\left(-(k-\ui\epsilon)^2\right)+\frac{d(k)}{2}\\
    &=N^W(k)+\lim_{\epsilon\rightarrow0^+}\frac{1}{\pi}\im\ln Z(\ui k(1-\ui\epsilon))+\frac{d(k)}{2},\quad k\in\rz_{>0},
}
\ee
where $Z(k)$ is defined in $\eref{r8}$ 
and $d(k)$ is defined in $\eref{q11d}$.
\end{theorem}
%
%
%
%
%
%
%
%
%
By the argumentation leading to theorem \ref{ra1} the series for $T^{po}(k)$  in \eref{r2} converges locally uniformly on $(\alpha,\infty)$, $\alpha\in\rz_{>0}$, and the improper Riemann integral in \eref{r8} converges uniformly with respect to the summation. Thus, integration and summation can be interchanged and we obtain for the principal branch of the logarithm, \cite[pp.\,379,380]{Abramowitz:1992} using the Schwarz reflection principle ($\lfp=\frac{2\pi}{n}$) 
\be
\label{r7}
\fl\eq{
\ln Z(s)&=-8\pi\sum\limits_{n=1}^{\infty}d(n)\int\limits_{s}^{\infty}\ker\left(4\pi\sqrt{nk}\right)\ud k\\
              &=-\frac{1}{\pi}\sum\limits_{n=1}^{\infty}\frac{d(n)}{n}\int\limits_{4\pi\sqrt{ns}}^{\infty}k'\ker k'\ud k'\\
              &=-2\sqrt{2}\sqrt{s}\sum\limits_{n=1}^{\infty}\frac{d(n)}{\sqrt{n}}\left[\ker_1\left(4\pi\sqrt{ns}\right)-\kei_1\left(4\pi\sqrt{ns}\right)\right]\\
              &=-4\sqrt{s}\sum\limits_{n=1}^{\infty}\frac{d(n)}{\sqrt{n}}\re\left[\ue^{-\ui\frac{\pi}{4}}K_1\left(4\pi\sqrt{ns}\ue^{\ui\frac{\pi}{4}}\right)\right]\\
              &=-\sqrt{\frac{2}{\pi}}\sqrt{s}\sum\limits_{n=1}^{\infty}d(n)\sqrt{\lfp}\left[\ue^{-\ui\frac{\pi}{4}}K_1\left(2(2\pi)^{\frac{3}{2}}\sqrt{\frac{s}{\lfp}}\ue^{\ui\frac{\pi}{4}}\right)\right.\\
              &\hspace{5.5cm}+\left.\ue^{\ui\frac{\pi}{4}}K_1\left(2(2\pi)^{\frac{3}{2}}\sqrt{\frac{s}{\lfp}}\ue^{-\ui\frac{\pi}{4}}\right)\right],\quad s>0,
}
\ee
where $\ker_1(z)$ and $\kei_1(z)$ denote the Kelvin functions. On account of \eref{rr3} and \eref{9a} the last series in \eref{r7} converges absolutely and locally uniformly on 
\be
\label{ar7}
\mathcal{S}:=\left\{s\in\kz,\quad \re s>0\right\}
\ee
and defines therefore an analytic function on $\mathcal{S}$.
%
%
%
%
Note that the penultimate line in \eref{r7} is not an analytic function on $\mathcal{S}$, it only coincides with $\ln Z(s)$ on $\rz_{>0}$. We evaluate the representation \eref{r7} for $\frac{1}{\pi}\ln Z(s)$ on the critical line $\ui\rz$. By the relation \cite[pp.\,66,67]{Magnus:1966} 
\be
\label{r24}
K_1\left(z\ue^{\frac{\ui\pi}{2}}\right)=-\frac{\pi}{2}\left(J_1(z)-\ui Y_1(z)\right),\quad z\in\kz,
\ee
we get
\be
\label{r25a}
\hspace{-2.5cm}\eq{
&\left.
\sqrt{s}\left[\ue^{-\frac{\ui\pi}{4}}K_1\left(4\pi\sqrt{ns}\ue^{\frac{\ui\pi}{4}}\right)+\ue^{\frac{\ui\pi}{4}}K_1\left(4\pi\sqrt{ns}\ue^{-\frac{\ui\pi}{4}}\right)\right]\right|_{s=\ui k}\\
& \hspace{2.0cm}=\frac{\sqrt{k}}{2}\left\{(1+\ui)\left[(1-\ui)K_1\left(4\pi\sqrt{nk}\ue^{\frac{\ui\pi}{2}}\right)+(1+\ui)K_1\left(4\pi\sqrt{nk}\right)\right]\right\}\\
&\hspace{2.0cm}=\sqrt{k}\left\{-\frac{\pi}{2}J_1\left(4\pi\sqrt{nk}\right)+\ui\left[ K_1\left(4\pi\sqrt{nk}\right)+\frac{\pi}{2}Y_1\left(4\pi\sqrt{nk}\right)\right]\right\},\quad k\in\rz_{>0}.
}
\ee
By inserting \eref{r25a} in \eref{r7} we formally get for the real and imaginary part of $\frac{1}{\pi}\ln Z(\ui k)$ on the critical line
\be
\label{r25b}
\fl\eq{
&\frac{1}{\pi}\re\ln Z(\ui k)=\frac{\sqrt{2}}{\pi^{\frac{3}{2}}}\sqrt{k}\sum\limits_{n=1}^{\infty}d(n)\sqrt{\lfp}J_1\left(2(2\pi)^{\frac{3}{2}}\sqrt{\frac{k}{\lfp}}\right),\\
&\frac{1}{\pi}\im\ln Z(\ui k)=\\
&\hspace{0.1cm}-\frac{\sqrt{2}}{\pi^{\frac{3}{2}}}\sqrt{k}\sum\limits_{n=1}^{\infty}d(n)\sqrt{\lfp}\left[K_1\left(2(2\pi)^{\frac{3}{2}}\sqrt{\frac{k}{\lfp}}\right)+\frac{\pi}{2}Y_1\left(2(2\pi)^{\frac{3}{2}}\sqrt{\frac{k}{\lfp}}\right)\right], \ k\in\rz_{>0}.
}
\ee
Due to the asymptotics \eref{rr3} and \cite[p.\,139]{Magnus:1966} ($|z|\rightarrow\infty$)
\be
\label{r25c}
\hspace{-2cm} J_\nu(z)=\left(\frac{1}{2}\pi z\right)^{-\frac{1}{2}}\cos\left(z-\frac{1}{2}\pi\nu-\frac{1}{4}\pi\right)\left(1+\Or\left(\frac{1}{z}\right)\right),\quad -\pi<\arg z<\pi,
\ee
and similarly as in section \ref{9} for $N^{\mathrm{Osc.}}(x)$ respectively $N_{\lead}^{\mathrm{Osc.}}(x)$, we can split up each series in \eref{r25b} in a normally convergent series on $k\in\rz\setminus\{0\}$ and a series which is given by (ignoring some constant factors which can be extracted from the summands) the real respectively the imaginary part of the Dirichlet series
\be
\label{r25d}
\sum\limits_{n=0}^{\infty}\frac{d(n)}{n^{\frac{3}{4}}}\ue^{-\ui4\pi\sqrt{ns}},
\ee    
with $s=k\in\rz\setminus{\gz}$. Hardy \cite{Hardy:1917} has proved that this series is locally uniformly convergent on $k\in\rz\setminus\gz$, thus, the two series in \eref{r25b} exist and define a continuous function on $\rz\setminus{\gz}$. In addition, the second series in \eref{r25b} converges even for all $k\in\rz\setminus\left\{0\right\}$ (see \eref{11d}). Note, that due to \eref{rr3} in the same manner as before for \eref{r25b} the series in the last line of the r.h.s. in \eref{r7} can be split up (ignoring again some constant factors which can be extracted from the summands) in a normally convergent series on 
\be
\label{g11}
\overline{\mathcal{S}}:=\left\{s\in\kz,\quad \re s\geq0\right\}
\ee
which is therefore holomorphic on $\mathcal{S}$ and locally uniformly continuous on $\overline{\mathcal{S}}$ and a Dirichlet series which is given by \eref{r25d} but now with $s\in\mathcal{S}$. For the next theorem we need a theorem which can be found in \cite[p.\,6]{Riesz:1915}:
\begin{theorem}[\cite{Riesz:1915}]
\label{sr1}
If a Dirichlet series $f$ is convergent for $s=s_0$, and has the sum $f\left(s_0\right)$, then $f(s)\rightarrow f\left(s_0\right)$ when $s\rightarrow s_0$ along any path which lies entirely within the region $\left|\arg\left(s-s_0\right)\right|\leq a<\frac{1}{2}\pi$ (the Dirichlet series is convergent for s$>s_0$).
\end{theorem}
Using theorem \ref{sr1}, the analyticity of $\ln Z(s)$ on $\mathfrak{S}$, the uniqueness of the analytic continuation, and the above results we infer the following theorem
\begin{theorem}
\label{r25aa}
\begin{enumerate}

 \item[i)]  $Z(s)$, defined in $\eref{r8}$, can be written as the following product over the classical primitive periodic orbits of the graph $\mathfrak{G}$
 \be
 \label{r25aa1}
\fl \eq{
& Z(s)=\prod\limits_{n=1}^{\infty}\exp\left(-\sqrt{\frac{2}{\pi}}\sqrt{s}\sum\limits_{m=1}^{\infty}\sqrt{\frac{\lfp}{m}}\left[\ue^{-\ui\frac{\pi}{4}}K_1\left(2(2\pi)^{\frac{3}{2}}\sqrt{\frac{ms}{\lfp}}\ue^{\ui\frac{\pi}{4}}\right)\right.\right.\\
&\hspace{5cm}\left.\left.+\ue^{\ui\frac{\pi}{4}}K_1\left(2(2\pi)^{\frac{3}{2}}\sqrt{\frac{ms}{\lfp}}\ue^{-\ui\frac{\pi}{4}}\right)\right]\right),\quad s\in\overline{\mathcal{S}}\setminus\left\{\ui\gz\right\}.
}
 \ee
\item[ii)] $\ln Z(s)$ is an analytic function on $\mathfrak{S}$ and satisfies
\be
\label{r25}
\fl\eq{
&\frac{1}{\pi}\ln Z(s)=\\
&\hspace{0.5cm}-\frac{\sqrt{2}}{(\pi)^{\frac{3}{2}}}\sqrt{s}\sum\limits_{n=1}^{\infty}d(n)\sqrt{\lfp}\left[\ue^{-\ui\frac{\pi}{4}}K_1\left(2(2\pi)^{\frac{3}{2}}\sqrt{\frac{s}{\lfp}}\ue^{\ui\frac{\pi}{4}}\right)\right.\\
&\hspace{5.0cm}+\left.
\ue^{\ui\frac{\pi}{4}}K_1\left(2(2\pi)^{\frac{3}{2}}\sqrt{\frac{s}{\lfp}}\ue^{-\ui\frac{\pi}{4}}\right)\right],\quad s\in\overline{\mathcal{S}}\setminus\left\{\ui\gz\right\}.
}%
\ee
\item[iii)] The real and the imaginary parts of $\frac{1}{\pi}\ln Z(s)$ on the critical line $s=\ui k$, $k\in\rz\setminus\gz$, are given by $\eref{r25b}$, in particular it holds, $k\in\rz_{>0}$,
\be
\label{r25aar}
\fl\eq{
N^{\mathrm{Osc.}}(k)&=\lim_{\epsilon\rightarrow0^+}\frac{1}{\pi}\im\ln Z(\ui k(1-\ui\epsilon))+\frac{d(k)}{2}\\
&=-\frac{\sqrt{2}}{\pi^{\frac{3}{2}}}\sqrt{k}\sum\limits_{n=1}^{\infty}d(n)\sqrt{\lfp}\left[K_1\left(2(2\pi)^{\frac{3}{2}}\sqrt{\frac{k}{\lfp}}\right)\right. \\
&\hspace{5.5cm} \left.+\frac{\pi}{2}Y_1\left(2(2\pi)^{\frac{3}{2}}\sqrt{\frac{k}{\lfp}}\right)\right]+\frac{d(k)}{2},
}
\ee
where $d(k)$ is defined in $\eref{q11d}$.
\end{enumerate}
\end{theorem}
Thus, the identity \eref{r25} justifies that the limit process $\epsilon\rightarrow 0$ in \eref{r22}, and the corresponding limit process on the r.h.s. in \eref{r25} can be interchanged with the summation. Note that the series in \eref{r25} is not absolutely convergent on the critical line $\ui\rz$, in contrast to the case $\re s>0$, where the series \eref{r25} is due to the asymptotics \eref{rr3}, in particular the exponential  damping term therein, absolutely  and locally uniformly convergent. Therefore, \eref{r25} can be considered as an adiabatic and analytic regularization of $N^{\mathrm{Osc.}}(x)$ in \eref{11d}.
\section{The asymptotics of $\ln \det\left(-\Delta+s^2\right)$ for $s\rightarrow 0$ and $s\rightarrow\infty$}
\label{500aaa}
In order to derive the asymptotics of $\ln \det\left(-\Delta+s^2\right)\equiv\ln \Df\left(s^2\right)$ for $s\rightarrow0$, we use, due to theorem \ref{63}, the representation \eref{64}. 
%
%
%
%
We then obtain for $\ln \Df\left(s^2\right)$, $|s|<1$,
\be
\label{1101}
\eq{
\ln \Df\left(s^2\right) &=-\ln(2\pi)+\sum\limits_{n=1}^{\infty}d(n)\ln\left(1+\frac{s^2}{n^2}\right)\\
         &=-\ln(2\pi)+\sum\limits_{m=1}^{\infty}\frac{(-1)^{m+1}}{m}\left[\sum\limits_{n=1}^{\infty}\frac{d(n)}{n^{2m}}\right]s^{2m}\\
         &=-\ln(2\pi)+\sum\limits_{m=1}^{\infty}\frac{(-1)^{m+1}\zeta^2(2m)}{m}s^{2m}\\
         &=-\ln(2\pi)+\left(\frac{\pi^2}{6}\right)^2s^2+\Or\left(s^4\right)\\
         &=\ln \Df(0)+T(0)s^2+\Or\left(s^4\right), \quad s\rightarrow0,
}
\ee
where we have used the absolute convergence of the double series in the second line in order to interchange the order of summation (in the third line we have used \eref{101} and in the last line \eref{300}).

By \eref{r10} we deduce the relation
\be
\label{w1}
\ln\Df\left(s^2\right)=\pi s\ln s+(2\gamma-1)\pi s+\frac{1}{2}\ln s+\ln Z(s),\quad s>0.
\ee
Using \eref{r22} and \eref{rr3} we can estimate $\ln Z(s)$ similarly as $T^{\po}(k)$ in section \ref{32aaa} as, $s\rightarrow\infty$,
\be
\label{aq1}
\fl\eq{
&\ln Z(s)=\\
&\hspace{0.5cm}-\sqrt{2}s^{\frac{1}{4}}\sum\limits_{n=0}^{\infty}\frac{d(n)}{n^{\frac{3}{4}}}\exp\left(-2\pi\sqrt{2ns}\right)\cos\left(2\pi\sqrt{2ns}+\frac{3\pi}{8}\right)\left(1+\Or\left(\frac{1}{\sqrt{ns}}\right)\right).
}
\ee
Again, the very first summand in the series \eref{aq1} determines the leading term of the asymptotics of $\ln Z(s)$ for $s\rightarrow\infty$, thus we can infer the following theorem:
\begin{theorem}
\label{508}
The logarithm of the functional determinant $\eref{62}$ (see theorem $\ref{63}$) possesses the asymptotics $\eref{1101}$ for $s\rightarrow0$ and
\be
\label{509}
\hspace{-2.4cm}\eq{
\ln \det\left(-\Delta+s^2\right) &=\pi s\ln s+\pi(2\gamma-1)s+\frac{1}{2}\ln s\\
&\hspace{-2.5cm} -\sqrt{2}s^{\frac{1}{4}}\exp\left(-2\pi\sqrt{2s}\right)\cos\left(2\pi\sqrt{2s}+\frac{3\pi}{8}\right)\left(1+\Or\left(\frac{1}{\sqrt{s}}\right)\right)\quad \mbox{for}\quad  s\rightarrow\infty.
}
\ee
\end{theorem}
We remark that the ``Weyl term'' of the asymptotic estimate \eref{509} agrees, of course, with the three leading terms in \eref{509a} derived by using the method of \cite{Sarnak:1987} and the trace formula \eref{13} for the quantum graph respectively by an evaluation of ${\Zl}(s,k)$ at $s=1$. 
\section{Conclusion}
\label{g1}
In this paper we studied in detail an infinite quantum graph $\mathfrak{G}$ which is the limit in the strong resolvent sense of a one-parameter family of Laplacians $(\Delta,\Dz_\kappa(\Delta))$, $\kappa>0$, acting on an infinite chain of edges with edge lengths $l_n=\frac{\pi}{n}$. This one-parameter family of Laplacians interpolates between the Kirchhoff Laplacian corresponding to $\kappa=0$ having a purely continuous spectrum and the Dirichlet Laplacian corresponding to $\kappa=\infty$ possessing a purely discrete spectrum. We proved that for each parameter $\kappa>0$ the operator $(\Delta,\Dz_\kappa(\Delta))$ possesses a purely discrete spectrum. 

The infinite Dirichlet quantum graph $\mathfrak{G}$ is characterized by pure Dirichlet boundary conditions at the vertices. The corresponding classical system describes a single particle moving freely on the edges and experiencing a hard wall reflection at the vertices. It describes therefore a regular classical system. Since the spectral multiplicities are given by the divisor function $d(n)$, this quantum graph is directly related to the well-known and still unsolved Dirichlet divisor problem, and thus the spectral counting function $N(x)$ coincides with the divisor summatory function. The Vorono\"i summation formula \eref{13} plays the role of an exact trace formula for this system. For the Dirichlet quantum graph we derived explicit formulae for the trace of the wave group $\Th(t)$, the trace of the generalized heat kernel ${\T}(t,k)$, the trace of the resolvent $T(k)$ and investigated their asymptotic behaviour for the arguments tending to zero or infinity. Furthermore, we examined various spectral zeta functions, $\Zh(s,k)$, $\tZh(s,k)$ and ${\Zl}(s,k)$, which can be used to define the spectral determinant $\det(-\Delta+k^2)$ in the sense of \cite{Hawking:1977}. For the spectral determinant $\det(-\Delta+k^2)$ we derived the asymptotics for $k\rightarrow\infty$ and calculated the Taylor series for $|k|<1$. These calculations are based on the investigation of a Selberg-like zeta function $Z(s)$ which is defined in terms of the periodic orbits of the system. We derived explicit expressions for $Z(s)$ valid on various domains in $\kz$ and established a connection with the oscillatory part $N^{\mathrm{Osc.}}(x)$ (remainder of the Dirichlet summatory function) of $N(x)$. Furthermore, we obtained a functional equation for $Z(s)$ which resembles that of the famous Selberg zeta function for the hyperbolic Laplacian acting on compact Riemann surfaces. Most interestingly $Z(s)$ fulfills the analogue of the Riemann hypothesis.
\subsection*{Acknowledgments}
We are very grateful to  Robin Nittka and the anonymous referees for very useful hints. S. E. would like to thank the graduate school ``Analysis of complexity, information
and evolution'' of the Land Baden-Württemberg for the stipend which has enabled this paper.
\\
\appendix
{\small
\bibliographystyle{unsrt}
\bibliography{litver1}}
\parindent0em
\end{document}